\documentclass[12pt,notitlepage,a4paper]{article}
\usepackage{amsmath,amssymb,amsthm}
\usepackage[bookmarks=false,colorlinks=false,pdfborder={0 0 0}]{hyperref}

\newtheorem{theorem}{Theorem}
\newtheorem{corollary}[theorem]{Corollary}
\newtheorem{lemma}[theorem]{Lemma}
\newtheorem{proposition}[theorem]{Proposition}
\newtheorem{definition}[theorem]{Definition}
\newtheorem{remark}{Remark}

\begin{document}

\title{\textbf{Large player classes and approximately weighted voting}}
\author{Ezra Einy\thanks{%
Department of Economics, Ben-Gurion University of the Negev, Beer Sheva
84105, Israel. E-mail: einy@bgu.ac.il} and Ori Haimanko\thanks{%
Department of Economics, Ben-Gurion University of the Negev, Beer Sheva
84105, Israel. E-mail: orih@bgu.ac.il}}
\date{\today }
\maketitle

\begin{abstract}
It is known that completeness of the coalitional ($C$-)desirability relation
in a simple game does not imply the existence of a weighted-majority
representation. We show that weightedness is implied in an\emph{\ approximate%
} sense. Call an $n$-player simple game $\varepsilon $\emph{-weighted} if,
for some weight vector in the $n-1$-simplex, any losing coalition weight
exceeds any winning coalition weight by at most $\varepsilon ;$ let $%
\varepsilon ^{\ast }$ be the smallest such $\varepsilon .$ We establish an
upper bound on $\varepsilon ^{\ast }$ in a general $C$-complete game, and
deduce that, if all classes of equally desirable players become large
asymptotically but the number of classes is bounded, then $\varepsilon
^{\ast }=$ $O(1/n_{\min })$ for the minimal player class size $n_{\min },$
and the bound is tight.

Importantly, the minimal attainable fraction $\delta $ of coalitions
mislabeled by a weighting rule also tends to $0$ when $n_{\min }\rightarrow
\infty $ and the number of player classes is bounded, with $\delta =$ $O(1/%
\sqrt{n_{\min }}).$

\emph{JEL} \emph{Classification Numbers}: C71, D72.

\end{abstract}

\section{Introduction}

\subsection{The background}

Deciding which simple games are weighted majority games is an old problem
(Isbell \hyperlink{cite.isbell56}{(1956}, \hyperlink{cite.Isbell58}{1958)},
Elgot \hyperlink{cite.Elgot60}{(1960)}, Muroga \hyperlink{cite.Muroga71}{%
(1971)}). The sufficient conditions for weightedness that are known, such as
asummability or trade robustness (Elgot \hyperlink{cite.Elgot60}{(1960)},
Taylor and Zwicker \hyperlink{cite.TZ92}{(1992)}), strong regularity (Einy
and Lehrer \hyperlink{cite.EL89}{(1989)}), invariant-trade robustness
(Freixas and Molinero \hyperlink{cite.FM09}{(2009)}), are by their nature
combinatorial and complex\footnote{%
For instance, trade robustness of Taylor and Zwicker \hyperlink{cite.TZ92}{%
(1992)} has a major combinatorial depth:\ an $n$-player game is weighted iff
it is $k$-trade robust for $k=(n+1)n^{n/2}$ (Gvozdeva and Slinko \hyperlink{cite.GS11}%
{(2011}, Theorem 4), and $k$ cannot in general be taken smaller than a
linear function of $n$ (ibid., Corollary 3). By contrast, the condition that
is at the core of this paper --- completeness of the coalitional
desirability relation --- is equivalent to $2$-trade robustness.}. A
conceptually elegant and, in a sense, ordinal feature of a weighted majority
game, namely, the completeness of Isbell's \hyperlink{cite.Isbell58}{(1958)}
desirability relation that the game induces over the player set ($I$-\emph{%
completeness}) --- or, more broadly, the completeness of Lapidot's
\hyperlink{cite.Lapidot68}{(1968)} desirability relation over the set of all
coalitions ($C$-\emph{completeness}) --- is a necessary but regrettably not
a sufficient condition for the weightedness of a general simple game.
Examples of $C$-complete non-weighted simple games have been reported in the
literature, starting with Einy \hyperlink{cite.Einy85}{(1985)} who exhibited
a non-weighted proper $9$-player game that has a complete coalitional
desirability relation with an acyclic strict part; Taylor and Zwicker
\hyperlink{cite.TZ95}{(1995}, \hyperlink{cite.TZ96}{1996)} give further
examples derived from a magic square.

In the setting with a continuum of nonatomic players, by contrast, Einy and
Neyman \hyperlink{cite.EN90}{(1990)} showed that $C$-completeness \emph{does}
fully characterize weighted majority games under certain continuity
assumptions. Here we attempt to capture a certain ``large game'' aspect that,
together with $C$-completeness, allows to obtain an asymptotic version of
weightedness in finite simple games.

\subsection{Approximate weightedness and the main results}

Most finite $C$-complete counterexamples to weightedness are small in a
specific way: the equivalence classes of the players w.r.t. the desirability
relation, i.e., maximal sets of equally desirable players, are tiny. In
Einy's two games the class sizes are $(1,2,3,2,1)$ and $(1,1,2,3,2,1)$; in
the magic-square game of Taylor and Zwicker \hyperlink{cite.TZ96}{(1996},
Section 4) all nine classes are singletons. Our corresponding positive
statement will be based on the following general result. Call a simple $n$%
-player game $\varepsilon $\emph{-weighted} if, for some weight vector in
the standard $(n-1)$-simplex, any losing coalition weight exceeds any
winning coalition weight by at most $\varepsilon ;$ also, let $\varepsilon
^{\ast }$ be the minimal value of $\varepsilon $ for which the game is $%
\varepsilon $-weighted. Thus, $\varepsilon ^{\ast }$ is a threshold: a game
is $\varepsilon $-weighted iff $\varepsilon \geq \varepsilon ^{\ast }$; a
game is (fully) weighted iff $\varepsilon ^{\ast }<0.$ For a $C$-complete
game, denote by $n_{\min }$ the size of its smallest desirability class (or,
equivalently, smallest symmetry class)\footnote{%
Our Proposition \ref{lem:orbits} shows that the symmetry classes in any $I$%
-complete (and therefore in any $C$-complete) game are precisely the game's $%
I$-desirability equivalence classes.} of players and by $d$ the number of
classes. Theorem \ref{thm:main} in Section \ref{upper bound section} states
that the loss-win-excess threshold $\varepsilon ^{\ast }$ has the following
upper bound:
\begin{equation*}
\varepsilon ^{\ast }\leq \;\left( \log _{2}d+2\right) /n_{\min },
\end{equation*}%
and in particular the game is $\varepsilon $-weighted for any $\varepsilon $
that exceeds that bound.

An obvious corollary reveals the ``large game'' counterpart of nonatomicity
that is conducive to weightedness: if each of the player classes grows with
the increase in the total electorate (that is, $n_{\min }\rightarrow \infty $%
) while the total number of classes $d$ does not explode (specifically, is
subexponential in $n_{\min }$), then increasing $n_{\min }$ forces \emph{%
approximate weightedness}: the threshold $\varepsilon ^{\ast }$ converges to
zero. If the number of classes is uniformly bounded, then the convergence
rate of $\varepsilon ^{\ast }$ is $O\!\left( 1/n_{\min }\right) ,$ and this
bound is shown to be asymptotically tight.

The assumption of growing player classes embodies the standard notion of
players becoming individually insignificant but with a correspondingly
growing number of their exact duplicates, and provides a quantitative bridge
between the finite model and the continuum model of Einy and Neyman
\hyperlink{cite.EN90}{(1990)}. In the latter work, $C$-completeness acted as
a ``convexifier'' of the set of winning coalitions (and midpoint-convexifier
of the set of the losing ones), after the two sets were embedded into a
convex structure of ideal (non-binary) coalitions; weights were obtained by
a hyperplane separation of the convex set from the convex hull of the
midpoint-convex one. Here, $C$-completeness plays a similar role, following
the step of projecting coalitions into their Euclidean \emph{profiles} via a
method introduced by Carreras and Freixas \hyperlink{cite.CF96}{(1996)} for $%
I$-complete games (whereby every coalition is mapped to the counting vector
of its players in each of the game's desirability classes). Because the
games are finite, $C$-completeness implies only approximate aspects of
(midpoint and general) convexity of the winning and the losing profile sets.
These aspects are identified (through a profile version of the $2$-trade
robustness property of Taylor and Zwicker \hyperlink{cite.TZ92}{(1992)} to
which $C$-completeness is equivalent) and exploited to obtain a bound on the
loss-win-excess threshold $\varepsilon ^{\ast }.$

\subsection{Alternative approaches}

Somewhat different approaches to approximate weightedness have been
considered in the literature, often with the aim to expand the domain of
games which may be deemed ``weighted'' and make these domains characterizable
by simply stated or elegant properties. Among these approaches is
pseudoweighting, whereby weights are assigned --- with various restrictions
--- to coalitions and not individual players (see, e.g., Taylor and Zwicker
\hyperlink{cite.TZ96}{(1996}, \hyperlink{cite.TZ99}{1999)} and Freixas
\hyperlink{cite.F21}{(2021)}). The closest in spirit to our $\varepsilon $%
-weighted game is a \emph{roughly weighted} game, introduced in Taylor and
Zwicker \hyperlink{cite.TZ96}{(1996}, \hyperlink{cite.TZ99}{1999)} and
extensively studied in Gvozdeva and Slinko \hyperlink{cite.GS11}{(2011)} and
Gvozdeva et al. \hyperlink{cite.GHS13}{(2013b)}, and $\alpha $-roughly
weighted games of Freixas and Kurz \hyperlink{cite.FK14}{(2014)}. Roughly
weighted games are a particular case of our $\varepsilon $-weighted games
when $\varepsilon =0;$ in a roughly weighted game that is not fully
weighted, under a best-fitting weight system there would always be
coalitions of both types with the same weight, and thus no weight vector
would fully discern all coalitions' winning status.

With $\alpha \geq 1$ being the maximal weight that a losing coalition is
allowed to carry while the winning coalitions are constrained to a minimal
weight of $1$, the class of $\alpha $-roughly weighted games of Freixas and
Kurz \hyperlink{cite.FK14}{(2014)} contains the games that are roughly
weighted in the strict sense (for which $\alpha =1$ when the rough weights
of all winning coalitions are positive). Our $\varepsilon $-weightedness is
related to their concept in the following way: Freixas and Kurz's \hyperlink{cite.FK14}%
{(2014)} critical value $\mu $ for the lowest $\alpha $ at which the game is
$\alpha $-roughly weighted is, in effect, obtained by minimizing (over
weight vectors) the \emph{ratio }between the maximal weight of a losing
coalition and the minimal weight of a winning coalition, while the threshold
$\varepsilon ^{\ast }$ is obtained by minimizing (over normalized weight
vectors) the \emph{difference} between the two. Because the underlying
objective functions are different, the two threshold values could, in
principle, be attained at different weight systems, but a close link exists
for most games, subject to certain qualifications. The convergence of $\max
\{\varepsilon ^{\ast },0\}$ to zero is equivalent to the convergence of $\mu
$ to 1 --- and so our asymptotic weightedness result holds for both modes of
convergence --- whenever the minimal winning coalitions do not become
negligible under an optimal weight vector (which is always the case, e.g.,
for proper games).

One crucial difference between our works lies in the domain of games to
which the thresholds are applied. Freixas and Kurz \hyperlink{cite.FK14}{%
(2014)} consider $I$-complete games as one of their smaller domains, and
show that no general proximity to weightedness can be obtained because $\mu
=\Theta \left( \sqrt{n}\right) $ for a sequence of $n$-player $I$-complete
games with just two player classes. Our results, on the other hand, are
stated for a more restricted domain of $C$-complete games, where asymptotic
convergence to weightedness is obtainable under plausible assumptions.

\subsection{Misclassification measure}

A recent paper of Albareda-Sambola et al. \hyperlink{cite.amr}{(2025)}
defines four optimization-based measures of distance from weightedness, of
which one has a direct bearing on our approximate weightedness, both
conceptually and quantitatively. Their \emph{outlier value} of a game
represents, in essence, the minimal count of coalitions whose constraints
must be dropped to make the remaining part of the game fully weighted. We
define a closely related \emph{misclassification measure} $\delta $ of a
game, which is the minimal attainable fraction of coalitions in the game
that are attributed a wrong status by a weighted majority representation; it
is just a normalized form of the outlier value.

Since positive values of the loss-win-excess threshold $\varepsilon ^{\ast }$
imply that no weight system can properly describe the status of all
coalitions, no matter how small $\varepsilon ^{\ast }$ is, the fact that any
coalition mislabeled by weighting is numerically ``close'' to its correct
status does not fix the misclassification. Fortunately, the asymptotic
behavior of the misclassification measure $\delta $ when the player classes
become large reflects approximate weightedness. In the central case of a
bounded number of player classes, both $\varepsilon ^{\ast }$ and $\delta $
converge to $0$ simultaneously, and $\delta =O(1/\sqrt{n_{\min }}).$ Thus,
asymptotically, the least attainable fraction of coalitions affected by
misattribution of status becomes vanishingly small in parallel with the
weight-measured convergence to weightedness.

\subsection{This paper's structure}

The rest of this paper is organized as follows. Section \ref{section
preliminaries} sets up the framework. It recalls the individual and the
coalitional desirability notions and states some related results, in
particular characterizing $C$-completeness by two versions of $2$-trade
robustness. It then introduces the $\varepsilon $-weightedness along with
the loss-win-excess threshold $\varepsilon ^{\ast },$ and establishes the
basic characterization of the latter as the value of a minimax problem.
Section \ref{upper bound section} contains the two main results: it
establishes the general upper and lower bound on $\varepsilon ^{\ast }$ in $%
C $-complete games, and states the asymptotic result on the convergence of $%
\varepsilon ^{\ast }$ to zero when player classes become large. Section \ref%
{sec asymptotic tightness} constructs a sequence of games that demonstrates
tightness of the $O\!\left( 1/n_{\min }\right) $ convergence rate. In that
sequence, each game's $C$-desirability relation has an acyclic strict part,
providing an opportunity to observe that the acyclicity of a strict $C$%
-desirability relation does not suffice to ensure even rough weightedness of
a $C$-complete game. Section \ref{alpha-weightedness section} shows that our
asymptotic result can be stated, to a large extent, in terms of $\alpha $%
-rough weightedness. Section \ref{sec misclassification} introduces the
misclassification measure and proves that it becomes vanishingly small in
the central asymptotic cases with large player classes. Section \ref%
{concluding remarks} concludes, and presents some open questions. All proofs
are deferred to the Appendix, except for the main steps in the proof of
Theorem \ref{thm:main} and the proof of Theorem \ref{thm:types}.

\section{Preliminaries and basic notions}

\label{section preliminaries}

\subsection{Simple games and desirability relations}

A \emph{(simple) game} is a pair $G=(N,\mathcal{W})$, where $N$ is a finite
player set with $\left\vert N\right\vert =n$, and $\mathcal{W}\subseteq
2^{N} $ is the set of \emph{winning coalitions} with $N\in \mathcal{W}$, $%
\emptyset \notin \mathcal{W}$, which is up-closed: $A\in \mathcal{W}$, $%
A\subseteq B\Rightarrow $ $B\in \mathcal{W}$. Members of $\mathcal{L}%
=2^{N}\setminus \mathcal{W}$ are the \emph{losing} coalitions. The game is
\emph{proper} if $A\in \mathcal{W}$ $\Rightarrow N\setminus A\in \mathcal{L}$%
, \emph{strong} if $A\in \mathcal{L}$ $\Rightarrow N\setminus A\in \mathcal{W%
}$, and \emph{constant-sum} if it is both proper and strong.

The following are two frequently considered desirability relations induced
by a simple game$.$

\begin{definition}[\textbf{Individual desirability;} \emph{Isbell \hyperlink{cite.Isbell58}%
{(1958)}, Maschler and Peleg \hyperlink{cite.MP66}{(1966)}}]
For a given $G=(N,\mathcal{W}),$ the $I$-desirability relation on $N$ is
defined as follows, for any $i,j\in N:$
\begin{equation*}
i\succeq _{I}j\quad \Longleftrightarrow \quad \forall \,R\subseteq
N\setminus \{i,j\}\ \left( R\cup \{j\}\in \mathcal{W}\Rightarrow R\cup
\{i\}\in \mathcal{W}\right) .
\end{equation*}%
The symmetric part of $\succeq _{I}$ is denoted by $\sim _{I},$ and the
strict part of $\succeq _{I}$ is denoted by $\succ _{I}$.
\end{definition}

A game is $I$\emph{-complete} if the relation $\succeq _{I}$ is complete;
this is the meaning of a ``complete game''\ in Carreras and Freixas \hyperlink{cite.CF96}%
{(1996)}, Taylor and Zwicker \hyperlink{cite.TZ99}{(1999)} and Freixas and
Molinero \hyperlink{cite.FM09}{(2009)}. The relation $\succeq _{I}$ is
clearly transitive, and so is a total preorder on $N$ in an $I$-complete
game; the central part $\sim _{I}$ is an equivalence relation.

\begin{definition}[\textbf{Coalitional desirability;} \emph{Lapidot
\hyperlink{cite.Lapidot68}{(1968)}, Einy \hyperlink{cite.Einy85}{(1985)}}]
\label{def:coal} For a given $G=(N,\mathcal{W}),$ the $C$-desirability
relation on $2^{N}$ is defined as follows, for any $A,B\subseteq N$:
\begin{equation*}
A\succeq _{C}B\quad \Longleftrightarrow \quad \forall \,R\subseteq
N\setminus (A\cup B)\ \left( B\cup R\in \mathcal{W}\Rightarrow A\cup R\in
\mathcal{W}\right) .
\end{equation*}%
The symmetric part of $\succeq _{C}$ is denoted by $\sim _{C},$ and the
strict part of $\succeq _{C}$ is denoted by $\succ _{C}$.
\end{definition}

A game is $C$\emph{-complete} if the relation $\succeq _{C}$ is complete;
this is what is meant by \textquotedblleft a game with a complete
desirability relation\textquotedblright\ in Einy \hyperlink{cite.Einy85}{%
(1985)}. Taking $A=\{i\}$, $B=\{j\}$ shows $\{i\}\succeq _{C}\{j\}$ iff $%
i\succeq _{I}j$, and so a $C$-complete game is also $I$-complete, but the
converse may fail: an $8$-player game in Einy \hyperlink{cite.Einy85}{(1985}%
, Example 4.6) has all players $\succeq _{I}$-ordered linearly, but some
coalitions are not $\succeq _{C}$-comparable.

Our main results assume $C$-completeness of a game. It is characterizable by
the property of $2$-trade robustness, due to Taylor and Zwicker \hyperlink{cite.TZ92}%
{(1992)}, defined below.

\begin{definition}[$2$\textbf{-trade robustness; }\emph{Taylor and Zwicker
\hyperlink{cite.TZ92}{(1992)}}]
\label{prop 2-trade robustness}A $2$\emph{-trade} of $G=(N,\mathcal{W})$ is
a quadruple $(X_{1},X_{2};Y_{1},Y_{2})$ with $X_{1},X_{2}\in \mathcal{W}$
and $Y_{1},Y_{2}\in \mathcal{L}$ (repetitions allowed) that satisfies
\begin{equation}
\chi _{X_{1}}+\chi _{X_{2}}\;=\;\chi _{Y_{1}}+\chi _{Y_{2}}  \label{2-trade}
\end{equation}%
on $N,$ i.e.,\ every player belongs to the same number of $X_{k}$ and $Y_{k}$%
. The game is $2$\emph{-trade robust} if it has no $2$-trade.
\end{definition}

The equivalence between $C$-completeness and $2$-trade robustness follows in
a rather straightforward fashion from known results, but has never been
stated explicitly.\footnote{%
By Taylor and Zwicker \hyperlink{cite.TZ95}{(1995}, Theorem 4.2), $2$-trade
robustness is equivalent to the lack of $2$-cycles in another relation
introduced in Lapidot \hyperlink{cite.Lapidot68}{(1968)}, whereby $%
A>_{L}B\quad $\ $\Longleftrightarrow $ $\exists \,R\subseteq N\setminus
(A\cup B)$ $\left( \left( B\cup R\in \mathcal{L}\right) \wedge \left( A\cup
R\in \mathcal{W}\right) \right) .$ But $\succeq _{C}$-completeness is
obviously equivalent to $>_{L}$ being $2$-acyclic, and thus to $2$-trade
robustness. There is a further connection between $>_{L}$ and $\succeq _{C}:$
$>_{L}$ is (fully)\ acyclic if and only if $\succeq _{C}$ is complete and
its strict part $\succ _{C}$ is acyclic (Taylor and Zwicker \hyperlink{cite.TZ95}%
{(1995}, p. 76)).} We record the following for completeness.

\begin{proposition}[$C$-\textbf{completeness }$\Leftrightarrow $ $2$-\textbf{%
trade robustness}]
\label{lem:inc} A game is $C$-complete iff it is $2$-trade robust.
\end{proposition}

\subsection{Player classes}

\label{section player classes}

Here we fix notation and recall central facts --- due to Carreras and
Freixas \hyperlink{cite.CF96}{(1996)} --- pertaining to the structure of $I$%
-complete simple games. In any given $I$-complete game $G=(N,\mathcal{W})$
the $I$-desirability relation $\succeq _{I}$ is a total preorder, and so $N$
can be partitioned into equivalence classes that form a well-defined chain $%
N_{1}\succ N_{2}\succ \dots \succ N_{d}$ (where $N_{t}\succ N_{t+1}$ means
that $i\succ _{I}j$ whenever $i\in N_{t}$ and $j\in N_{t+1}$). These sets
will be referred to as \emph{(player) classes}. Denote by $n_{t}=|N_{t}|$
the size of the class $N_{t},$ and let $\bar{n}=(n_{1},\dots ,n_{d}).$ For
each coalition $A\subseteq N$ let
\begin{equation}
s(A)=\bigl(|A\cap N_{1}|,\dots ,|A\cap N_{d}|\bigr)\in \Lambda (\bar{n}%
):=\prod_{t=1}^{d}\{0,1,\dots ,n_{t}\}  \label{s(S)}
\end{equation}%
be the \emph{profile} of $A,$ specifying the count of its representatives in
each class. It is easy to see that any two coalitions have the same profile
if and only if one is carried to the other by a product of transpositions of
$\sim _{I}$-equivalent players.

The profile sets of the winning and the losing coalitions are denoted as
follows:
\begin{equation*}
\overline{\mathcal{W}}=\{\,s(A)\ :\ A\in \mathcal{W}\,\}\subseteq \Lambda (%
\bar{n}),\qquad \overline{\mathcal{L}}=\Lambda (\bar{n})\setminus \overline{%
\mathcal{W}}.
\end{equation*}%
Thus, $A\in \mathcal{W}$ iff $s(A)\in \overline{\mathcal{W}}$ and $A\in
\mathcal{L}$ iff $s(A)\in \overline{\mathcal{L}}.$ By the up-closedness of $%
\mathcal{W}$, $\overline{\mathcal{W}}$ is up-closed and $\overline{\mathcal{L%
}}$ is down-closed in $\Lambda (\bar{n})$ for the coordinatewise order.

To discuss the symmetries of $G,$ consider a permutation $\sigma $ of $N.$
The permutation $\sigma $ acts on $2^{N}$ in the following way: for every $%
A\subseteq N$, $\sigma (A)=$ $\left\{ \sigma (i):i\in A\right\} .$ It is
called an \emph{automorphism} of $G$ if $\sigma \left( \mathcal{W}\right) =%
\mathcal{W}$, i.e., the set of winning coalitions in the game is invariant
when the players are permuted according to $\sigma .$ Denote the set of all
automorphisms of $G$ by $\Gamma (G).$ According to the following lemma, $G$
has precisely $d$ symmetry classes (minimal nonempty sets preserved by all
automorphisms), which coincide with its $\sim _{I}$-equivalence classes $%
N_{1},N_{2},\dots ,N_{d}$.

\begin{proposition}[\textbf{Equivalence classes are symmetry classes}]
\label{lem:orbits} If $G$ is $I$-complete, $\Gamma (G)$ consists of all
permutations that preserve each $N_{t}$ (acting on $N_{t}$ in an arbitrary
way).
\end{proposition}

The following proposition is a profile counterpart of Proposition \ref%
{lem:inc}: it characterizes $C$-completeness in terms of $2$-trade
robustness stated for general coalition \emph{profiles}. This makes the $2$%
-trade robustness more straightforward to use as a technical tool in the
profile setting.

\begin{proposition}[\textbf{Profile }$2$-\textbf{trade robustness}]
\label{prop:sumset} Assume that $G=(N,\mathcal{W})$ is $I$-complete. Then it
is $C$-complete iff there are no $x,y\in \overline{\mathcal{W}}$ and $u,v\in
\overline{\mathcal{L}}$ with
\begin{equation}
x+y=u+v;  \label{sumset eq}
\end{equation}%
that is, iff $(\overline{\mathcal{W}}+\overline{\mathcal{W}})\cap (\overline{%
\mathcal{L}}+\overline{\mathcal{L}})=\emptyset $.
\end{proposition}

\subsection{Weightedness, $\protect\varepsilon $-weightedness and the
loss-win-excess threshold}

Denote by $\Delta \left( N\right) $ the set of probability vectors on $N.$

\begin{definition}[\textbf{Three versions of weightedness}]
\label{def:eps} Let $\varepsilon \in \mathbb{R}$. A simple game $G=(N,%
\mathcal{W})$ is $\varepsilon $-\emph{weighted} if there exist a \emph{%
weight vector} $w\in \Delta \mathbb{(}N\mathbb{)}$ and a \emph{quota} $q\in
\mathbb{R}$ such that
\begin{equation}
A\in \mathcal{W}\ \Longrightarrow \ w(A)\ \geq \ q\text{ and }B\in \mathcal{L%
}\ \Longrightarrow \ w(B)\ \leq \ q+\varepsilon ,  \label{epsilon weighted}
\end{equation}%
where $w(A)=\sum_{i\in A}w_{i}$. A game is \emph{roughly weighted if it is }$%
0$-weighted.\emph{\ }A game is (fully) \emph{weighted} if it is roughly
weighted and the second inequality in (\ref{epsilon weighted}) is strict: $%
B\in \mathcal{L}\ \Longrightarrow \ w(B)\ <q.$
\end{definition}

Notice that the definition is monotone: if $\varepsilon \leq \varepsilon
^{\prime }$ and the game is $\varepsilon $-weighted then it is $\varepsilon
^{\prime }$-weighted. Also, every game is $1$-weighted (take any $w$ and $%
q=0 $). And a game is weighted iff it is $\varepsilon $-weighted for some $%
\varepsilon <0$ --- the absolute value of a negative $\varepsilon $
describes a (positive) gap between the maximal weight of a losing coalition
and the minimal weight of a winning coalition. It is easy to see that any
weighted game is $C$-complete.

The definition of a roughly weighted game is due to Taylor and Zwicker
\hyperlink{cite.TZ95}{(1995}, \hyperlink{cite.TZ99}{1999)}. Clearly, every
weighted game is roughly weighted (by the same $w,$ $q$), but the converse
is not true: there are roughly weighted games in which \emph{any} weight
vector satisfying the requirement in (\ref{epsilon weighted}) for $%
\varepsilon =0$ would necessarily give the same weight to some winning and
losing coalitions (see, e.g., Gvozdeva and Slinko \hyperlink{cite.GS11}{%
(2011)} for a six-player example).

The following definition casts the search for the minimal $\varepsilon $ for
which (\ref{epsilon weighted}) holds as a minimax problem.

\begin{definition}[\textbf{Loss-win-excess threshold}]
\label{de loss win gap}Given $G=(N,\mathcal{W}),$ define its \emph{%
loss-win-excess threshold }as%
\begin{equation}
\varepsilon ^{\ast }(G)=\min_{w\in \Delta (N)}\max_{A\in \mathcal{W},B\in
\mathcal{L}}\left[ w(B)-w(A)\right] =\min_{w\in \Delta (N)}\left( \max_{B\in
\mathcal{L}}w(B)-\min_{A\in \mathcal{W}}w(A)\right) .  \label{epsilon **}
\end{equation}
\end{definition}

Thus, $\varepsilon ^{\ast }(G)$ is the least value that weighting can
achieve of the maximum weight difference between the losing and winning
coalitions. It constitutes a threshold for $\varepsilon $-weightedness: if $%
\varepsilon \geq \varepsilon ^{\ast }(G)$ then $G$ is $\varepsilon $%
-weighted (with a weight vector $w^{\ast }\in \Delta (N)$ attaining the
minimum in (\ref{epsilon **}) and a quota$\ q=$ $\min_{A\in \mathcal{W}%
}w^{\ast }(A)$), and $G$ is not $\varepsilon $-weighted whenever $%
\varepsilon <\varepsilon ^{\ast }(G).$ And, clearly, $G$ is (fully) weighted
iff $\varepsilon ^{\ast }(G)<0,$

Note that $\varepsilon ^{\ast }(G)$ is well-defined, being the minimax value
of a continuous function with a compact set of maximizers and of minimizers.
When $G$ is constant-sum, it is easy to see that (\ref{epsilon **}) boils
down to a simpler expression,%
\begin{equation}
\varepsilon ^{\ast }(G)=1-2\max_{w\in \Delta (N)}\min_{A\in \mathcal{W}}w(A).
\label{epsilon constant sum}
\end{equation}

The following lemma provides a computational formula for $\varepsilon ^{\ast
}$ in $I$-complete games, based on the minimax theorem. It will be of
crucial use in obtaining the upper estimate for the threshold $\varepsilon
^{\ast }$ in Theorem \ref{thm:main}.

\begin{proposition}[\textbf{Formula for }$\protect\varepsilon ^{\ast }$]
\label{lem:minimax} For an $I$-complete $G=(N,\mathcal{W}):$

\begin{itemize}
\item[(i)] The weight vector attaining the minimum in (\ref{epsilon **}) can
be chosen as \emph{constant} on each player class;

\item[(ii)] $\varepsilon ^{\ast }(G)$ is computable as follows:%
\begin{equation}
\varepsilon ^{\ast }(G)=\max \left\{ \min_{1\leq t\leq d}\frac{b_{t}-a_{t}}{%
n_{t}}\ :\ a\in \operatorname{conv}\overline{\mathcal{W}},\ b\in \operatorname{conv}%
\overline{\mathcal{L}}\right\} ,  \label{eq:mm}
\end{equation}%
where $\operatorname{conv}\overline{\mathcal{W}}\subset \mathbb{R}_{+}^{d}$ and $%
\operatorname{conv}\overline{\mathcal{L}}\subset \mathbb{R}_{+}^{d}$ denote the
convex hulls of $\overline{\mathcal{W}}$ and of $\overline{\mathcal{L}}$.
\end{itemize}
\end{proposition}

\begin{remark}[\textbf{On weightedness and trade robustness}]
We note in passing, but do not use, the fact that (full) weightedness of a
game --- which is a significantly stronger property than $C$-completeness
--- also possesses a characterization by means of a strengthened version of $%
2$-trade robustness in Definition \ref{prop 2-trade robustness}: for \emph{%
any} $k,$ there should exist no $k$-trade, where the equality (\ref{2-trade}%
) has $k$ summands with sets from $\mathcal{W}$ \ (on the left-hand side)
and $k$ summands with sets from $\mathcal{L}$ (on the right-hand side). That
is precisely the definition of (full) \emph{trade robustness}, due to Taylor
and Zwicker \hyperlink{cite.TZ92}{(1992)}.
\end{remark}

\section{The main results}

\label{upper bound section}

\subsection{Bounds on the loss-win-excess threshold}

\label{sec computing epsilon*}

The next theorem provides an upper and a lower bound on $\varepsilon ^{\ast
} $ in any $C$-complete game, and drives the rest of our results.

\begin{theorem}[\textbf{Bounds on the loss-win-excess threshold}]
\label{thm:main} Let $G$ be a $C$-complete game with $d$ player classes.
Denote by $n_{\min }=\min_{t}n_{t}$ the smallest player class size in the
game. Then
\begin{equation}
-1/n_{\min }\leq \varepsilon ^{\ast }(G)\text{ }\leq \left( \log
_{2}d+2\right) /n_{\min }.  \label{epsilon star G}
\end{equation}
\end{theorem}

Notice that the bounds in (\ref{epsilon star G}) are \emph{uniform} across
all $C$-complete $n$-player games that share the same $d$ and $n_{\min }.$%
\footnote{%
By reviewing the proof of Lemma \ref{lem:round} in Section \ref{sec lemma
round}, a slightly better but less elegant upper bound $\left( \lfloor \log
_{2}d\rfloor +1+\frac{d}{2^{\lfloor \log _{2}d\rfloor }}\right) /n_{\min }$
on $\varepsilon ^{\ast }(G)$ can be obtained. Establishing the lower bound
on $\varepsilon ^{\ast }(G)$ in (\ref{epsilon star G}) --- admittedly a less
important part of Theorem \ref{thm:main} --- only uses Proposition \ref%
{lem:minimax}, and so the lower bound applies not only to $C$-complete, but
to $I$-complete games in general.} Since $\varepsilon ^{\ast }(G)\leq 1$ for
any game $G$, as is immediate from (\ref{epsilon **}), the upper bound $%
\left( \log _{2}d+2\right) /n_{\min }$ in (\ref{epsilon star G}) is only
informative when it is small. The latter is the case when the minimal class
size $n_{\min }$ is large, and $\log _{2}d$ is small by comparison, which
will be exploited in Theorem \ref{thm:types} to establish the smallness of $%
\varepsilon ^{\ast }(G)$ --- or closeness to weightedness --- for games with
large player classes.

Importantly, the upper bound in (\ref{epsilon star G}) on $\varepsilon
^{\ast }$ is \emph{positive}, while the published examples of $C$-complete
non-weighted games of which we are aware are\emph{\ roughly} (or $0$-)\emph{%
weighted}, i.e., have $\varepsilon ^{\ast }=0.$ Such are the two $9$- and $%
10 $-player games of Einy \hyperlink{cite.Einy85}{(1985)} with class-size
vectors $\overline{n}=$ $(1,2,3,2,1)$ and $\overline{n}=$ $(1,1,2,3,2,1)$,
and the $11$-player $C$-complete games of Freixas et al. \hyperlink{cite.FFK17}%
{(2017}, p. 493) with class sizes $\overline{n}=$ $(3,3,5),$ for which $%
\varepsilon ^{\ast }=0$ is verifiable computationally. In a large and
complex family introduced in Taylor and Zwicker \hyperlink{cite.TZ95}{(1995)}
of strongly rigid $k\times k$ magic square-games for $k\geq 3$ (with $k^{2}$
players), their slivered variants (with $k^{2}+k^{2}(k-1)^{2}/2$ players)
and their constant-sum extensions, all games are roughly weighted by
construction, non-weighted and $C$-complete (by Taylor and Zwicker's
\hyperlink{cite.TZ95}{(1995}, Lemma 2.2, Theorem 3.4, Theorem 4.4(2)--(3)
and Lemma 5.9)).

Could it be the case that the upper bound in Theorem \ref{thm:main}
overestimates $\varepsilon ^{\ast }(G)$? Certainly not in terms of sign or
asymptotic order. In Section \ref{sec asymptotic tightness} we show that the
bound is\emph{\ asymptotically tight }when the player classes become large;
in particular, in a family of $C$-complete games that we exhibit there, the
loss-win threshold $\varepsilon ^{\ast }$ is \emph{positive} for all
sufficiently large player classes.

\subsection{Proof of Theorem \protect\ref{thm:main}}

\label{sec:proof}

This section contains the main steps of the proof of Theorem \ref{thm:main};
the proofs of all internal lemmas as well as the derivation of the lower
bound in (\ref{epsilon star G}) appear in Section \ref{sec missing proof
components} of the Appendix. In all that follows, vectors in any Euclidean
space $\mathbb{R}^{d}$ are compared coordinatewise, $\mathbf{1\in }\mathbb{R}%
^{d}$ always denotes $(1,\dots ,1)$, and the ceiling and floor functions, $%
\lceil \cdot \rceil $ and $\lfloor \cdot \rfloor ,$ act coordinatewise.

\subsubsection{Approximate midpoint convexity}

It is the method of projecting coalitions into vector profiles that creates
the Euclidean framework in which the convex analysis can be used. The $2$%
-trade robustness characterization of $C$-completeness in Proposition \ref%
{prop:sumset} applies to profiles, and this allows to deduce some aspects of
convexity of the winning and the losing profile sets, $\overline{\mathcal{W}}
$ and $\overline{\mathcal{L}}.$

\begin{lemma}[\textbf{Integer midpoint convexity}]
\label{lem:mid} For a $C$-complete $G=(N,\mathcal{W}),$ let $x,y\in
\overline{\mathcal{W}}$ be such that $(x+y)/2\in \Lambda (\bar{n})$ (i.e.,
the average has integer coordinates). Then $(x+y)/2\in \overline{\mathcal{W}}
$. Symmetrically, if $x,y\in \overline{\mathcal{L}}$ and $(x+y)/2\in \Lambda
(\bar{n})$, then $(x+y)/2\in \overline{\mathcal{L}}$.
\end{lemma}

The next lemma shows that $\overline{\mathcal{W}}$ (respectively, $\overline{%
\mathcal{L}}$) is closed under upper (respectively, lower) rounding of
midpoints.

\begin{lemma}[\textbf{Rounded midpoint convexity}]
\label{lem:rmid} Let $G=(N,\mathcal{W})$ be $C$-complete. If $x,y\in
\overline{\mathcal{W}}$ then $\lceil (x+y)/2\rceil \in \overline{\mathcal{W}}
$; if $x,y\in \overline{\mathcal{L}}$ then $\lfloor (x+y)/2\rfloor \in
\overline{\mathcal{L}}$.
\end{lemma}

A comparison with the parallel components in the proof of Einy and Neyman's
\hyperlink{cite.EN90}{(1990)} weightedness result in the nonatomic continuum
framework is instructive. In that framework, the place of our profile set $%
\Lambda (\bar{n})$ is taken by the convex set of \emph{ideal} coalitions,
which are $[0,1]$-valued measurable functions on the player set that allow
players' fractional participation. Under their continuity assumptions in an
appropriate topology, ideal coalitions can be approximated by genuine
coalitions with the same winning status. In the vector-measure setting of
Einy and Neyman \hyperlink{cite.EN90}{(1990}, Corollary 4.2), they can be
matched exactly in terms of the defining vector measure.

Just as we do in Proposition \ref{prop:sumset}, Einy and Neyman \hyperlink{cite.EN90}%
{(1990}, Proposition 4.1 and Lemma 4.2) prove --- though implicitly --- a
version of $2$-trade robustness that $C$-completeness implies for genuine
coalitions and utilize it to show midpoint convexity of the winning and
losing \emph{ideal} sets via Lyapunov's convexity theorem. But this cannot
be done with our coalition profiles --- vectors in the ``ideal'' profiles set $%
\operatorname{conv}\Lambda (\bar{n})$ need not be the genuine \emph{integer}
profiles in $\Lambda (\bar{n})$ --- and so our variant of midpoint convexity
in Lemma \ref{lem:rmid} is necessarily rounded, not precise.

\subsubsection{Approximating convex combinations by pure profiles}

\label{approximating convex combinations}

According to Lemma \ref{lem:rmid}, $\overline{\mathcal{W}}$ is closed under
the averaging operation up to an error of at most $1/2$ per coordinate, and
a mirror claim holds for $\overline{\mathcal{L}}$ . By Dilworth et al.
\hyperlink{cite.DHR}{(1999}, Theorem 1), this type of approximate midpoint
convexity translates into a\emph{\ general }convexity with a uniform sharp
bound on errors. Our next lemma applies this result to $\overline{\mathcal{W}%
}$ and $\overline{\mathcal{L}}$: any convex combination of profiles in $%
\overline{\mathcal{W}}$ (respectively, in $\overline{\mathcal{L}}$)
dominates (respectively, is dominated by) a pure coalition profile with the
same winning status, provided a bounded error is allowed in the comparison
of coordinates.

\begin{lemma}[\textbf{Dominance with bounded error}]
\label{lem:round} For a $C$-complete $G=(N,\mathcal{W}),$ put
\begin{equation}
\gamma =\left( \log _{2}d+2\right) /2.  \label{delta}
\end{equation}%
Then for every $a\in \operatorname{conv}\overline{\mathcal{W}}$ there is $u\in
\overline{\mathcal{W}}$ with $u\leq a+\gamma \mathbf{1}$, and for every $%
b\in \operatorname{conv}\overline{\mathcal{L}}$ there is $z\in \overline{\mathcal{L}}
$ with $z\geq b-\gamma \mathbf{1}.$
\end{lemma}

To return briefly to the comparison with Einy and Neyman \hyperlink{cite.EN90}%
{(1990)}, their full midpoint convexity translates immediately into the
general convexity (at least for the set of ideal winning coalitions) due to
the density of dyadic fractions and the closedness of the ideal winning set.
Here, the extra step of converting approximate midpoint convexity to a
general approximate convexity (with a compound error term) is needed, which
is performed in Lemma \ref{lem:round} via a convenient off-the-shelf
approximation result of Dilworth et al. \hyperlink{cite.DHR}{(1999)}.

\subsubsection{Conclusion of the proof}

\label{conclusion of the proof of theorem 1}

Let $G=(N,\mathcal{W})$ be a $C$-complete game, and $\gamma $ be given by (%
\ref{delta}) in Lemma \ref{lem:round}. Suppose for contradiction that%
\begin{equation}
\varepsilon ^{\ast }(G)\geq 2\gamma /n_{\min }.  \label{contradiction delta}
\end{equation}%
Take $a\in \operatorname{conv}\overline{\mathcal{W}}$ and $b\in \operatorname{conv}\overline{%
\mathcal{L}}$ that attain the maximum in (\ref{eq:mm}) of Proposition \ref%
{lem:minimax}. Thus%
\begin{equation}
\varepsilon ^{\ast }(G)=\min_{1\leq t\leq d}\frac{b_{t}-a_{t}}{n_{t}}\ ,
\label{epsilon star = b_t-a_t/n_t}
\end{equation}%
so that

\begin{equation}
b_{t}\ \overset{\text{(\ref{epsilon star = b_t-a_t/n_t})}}{\geq }\
a_{t}+\varepsilon ^{\ast }(G)n_{t}\geq a_{t}+\varepsilon ^{\ast }(G)n_{\min }%
\overset{\text{(\ref{contradiction delta})}}{\geq }a_{t}+2\gamma
\label{b_t>=a_t}
\end{equation}%
for each $1\leq t\leq d$. By Lemma \ref{lem:round} there are $u\in \overline{%
\mathcal{W}}$ with $u\leq a+\gamma \mathbf{1}$ and $z\in \overline{\mathcal{L%
}}$ with $z\geq b-\gamma \mathbf{1}$. Thus,
\begin{equation*}
z\ \geq \ b-\gamma \mathbf{1}\ \overset{\text{(\ref{b_t>=a_t})}}{\geq }\
a+\gamma \mathbf{1}\ \geq \ u.
\end{equation*}%
Since $\overline{\mathcal{W}}$ is up-closed and $u\in \overline{\mathcal{W}}$%
, the above inequality forces $z\in \overline{\mathcal{W}}$. But this
contradicts $z\in \overline{\mathcal{L}}$ since $\overline{\mathcal{W}}$ and
$\overline{\mathcal{L}}$ are obviously disjoint. Hence,
\begin{equation*}
\varepsilon ^{\ast }(G)<2\gamma /n_{\min }=\left( \log _{2}d+2\right)
/n_{\min }.
\end{equation*}

The lower bound on $\varepsilon ^{\ast }(G)$ is established in Section \ref%
{sec lower bound} of the Appendix.

\subsection{Asymptotics with large player classes}

\label{increasingly large section}

For $C$-complete games, weightedness is guaranteed in some cases with a
small number of players, or with small number of player classes: it is known
that all $C$-complete games with $8$ players or less are weighted (by Muroga
et al. \hyperlink{cite.Muroga70}{(1970)}), and for three-class games the
same holds up to $10$ players (by Freixas et al. \hyperlink{cite.FFK17}{(2017%
}, p. 493));\footnote{%
The claims of these authors are stated for $2$-trade robust (or $2$%
-asummable) games, which are precisely the $C$-complete games by Proposition %
\ref{lem:inc}.} furthermore, all $C$-complete games with just two player
classes are weighted (see Remark \ref{Rem weightedness for all n}). The
interest in approximate weightedness therefore lies in large-$n$ games
(starting with $n=9$) with at least three player classes, when
non-weightedness begins to bite (as pointed out in the discussion following
Theorem \ref{thm:main}).

Our next result is asymptotic: as long as all player classes increase in
size, and the total number of player classes does not grow explosively, the
loss-win-excess threshold becomes vanishingly small.

\begin{theorem}[\textbf{Asymptotic smallness of }$\protect\varepsilon ^{\ast
}$]
\label{thm:types} Let $\left\{ G_{n}\right\} _{n\geq 1}$ be a sequence of $n$%
-player $C$-complete games in which $\lim_{n\rightarrow \infty }n_{\min
}=\infty $ and the number of player classes $d_{n}$ is subexponential in $%
n_{\min }:$ $d_{n}=\exp (o(n_{\min })).$ Then%
\begin{equation}
\lim_{n\rightarrow \infty }\varepsilon ^{\ast }\left( G_{n}\right) =0.
\label{epsilon goes to zero}
\end{equation}%
Furthermore, if $d_{n}$ is bounded along the sequence then
\begin{equation}
\varepsilon ^{\ast }\left( G_{n}\right) =\;O\!\left( 1/n_{\min }\right) .
\label{cpnvergence rate}
\end{equation}
\end{theorem}

\begin{proof}
By Theorem \ref{thm:main},
\begin{equation*}
-1/n_{\min }\leq \varepsilon ^{\ast }\left( G_{n}\right) \leq \left( \log
_{2}d_{n}+2\right) /n_{\min },\text{ }
\end{equation*}%
and so, by the assumptions on $n_{\min }$ and $d_{n},$ (\ref{epsilon goes to
zero}) holds. And, if $d_{n}\leq d$ for all $n,$ then
\begin{equation*}
\varepsilon ^{\ast }\left( G_{n}\right) \leq \left( \log _{2}d+2\right)
/n_{\min }\ \Longrightarrow \varepsilon ^{\ast }\left( G_{n}\right)
=O\!\!\left( 1/n_{\min }\right) .
\end{equation*}%
$\;$
\end{proof}

Note that the player classes need not be comparable in size in order for
Theorem \ref{thm:types} to hold, but with classes of comparable magnitude
the bound on the number of classes $d_{n}$ becomes automatic: if each player
class contains at least a fixed positive fraction of the electorate, that
is, $n_{\min }\geq $ $cn$ for some $c>0,$ then $n\geq $ $d_{n}n_{\min }\geq $
$d_{n}cn$ and therefore $d_{n}$ is bounded from above by $\frac{1}{c}$ for
all $n.$

\begin{remark}[\textbf{Weightedness for any} $n$]
\label{Rem weightedness for all n}For two sets of $I$-complete games, those
with just $d=2$ player classes and those that have a single shift-minimal
winning profile,\footnote{%
For a shift-minimal winning profile $x\in \overline{\mathcal{W}},$ if $y\in
\Lambda \left( \overline{n}\right) $ is such that $\sum_{t=1}^{k}y_{t}$ $%
\leq $ $\sum_{t=1}^{k}x_{t}$ for every $1\leq k\leq d$ with at least one
strict inequality, then $y\in \overline{\mathcal{L}}\,.$} $C$-completeness
alone is known to imply (full) weightedness for \emph{any} number of
players, without the need to resort to asymptotics. Indeed, $C$-completeness
is equivalent to $2$-trade robustness by Proposition \ref{lem:inc}, and the
latter implies its weaker variant, $2$-invariant-trade robustness.\footnote{%
The concept of $2$-invariant-trade robustness is due to Freixas and Molinero
\hyperlink{cite.FM09}{(2009)}. It only prohibits those $2$-trades in (\ref%
{2-trade}) at which the winning coalitions have shift-minimal winning
profiles.} By Freixas et al. \hyperlink{cite.FFK17}{(2017}, Theorems 5.3 and
5.9), any $I$-complete and $2$-invariant-trade robust game that belongs to
either of the two sets above is weighted. Non-weightedness in a $C$-complete
game may therefore only occur if there are at least three player classes
\emph{and} at least two shift-minimal winning profiles.
\end{remark}

\begin{remark}[\textbf{Insufficiency of }$I$\textbf{-completeness for}
\textbf{asymptotic weightedness}]
\label{I-completeness does not imply asymptotic weightedness}Freixas and
Kurz \hyperlink{cite.FK14}{(2014)} showed that the domain of $I$-complete
games extends far beyond weighted games, as captured by their
critical-threshold value that we present in Section \ref{alpha-weightedness
section}. Here we explicitly demonstrate that $\varepsilon ^{\ast }$ may be
bounded away from $0$ for $I$-complete games with two increasing player
classes. For any $m\geq 2,$ consider a $4m$-player game $\widehat{G}_{m}$
with player classes $N_{1}\succ N_{2}=N\diagdown N_{1}$ of equal size $%
n_{1}=n_{2}=2m,$ in which $A\in \mathcal{W}$ iff $|A\cap N_{1}|\geq $ $m+1$
and $|A|\geq 2m+1.$\footnote{%
Using the terminology of Gvozdeva et al. \hyperlink{cite.GHS11}{(2013a,
Definition 11)}, $\widehat{G}_{m}$ is a \emph{conjunctive hierarchical game}
with two seniority ranks.}Clearly, $\widehat{G}_{m}$ is proper and $I$%
-complete. Next, let $u=$ $(m,2m)$ $\in \overline{\mathcal{L}},$ $v=(2m,\,0)$
$\in \overline{\mathcal{L}}.$ It follows from (\ref{eq:mm}) in Proposition %
\ref{lem:minimax}(ii), by considering a fixed $a=(m+1,m)\in \overline{%
\mathcal{W}}$ and all $b=$ $\lambda u+(1-\lambda )v$ $\in \operatorname{conv}%
\overline{\mathcal{L}}$ for $0\leq \lambda $ $\leq 1,$ that%
\begin{equation}
\varepsilon ^{\ast }(\widehat{G}_{m})\geq \max_{0\leq \lambda \leq 1}\min
\left\{ \frac{\left( 1-\lambda \right) m-1}{2m},\frac{2\lambda -1}{2}%
\right\} =\frac{1}{6}-\frac{1}{3m}.  \label{G^m}
\end{equation}%
Accordingly,\footnote{%
In fact, the lower bound in (\ref{G^m}) is the exact value of $\varepsilon
^{\ast }(\widehat{G}_{m})$. Indeed, by setting $w_{i}=1/(3m)$ for each $i\in
N_{1}$ and $w_{i}=1/(6m)$ for each $i\in N_{2},$ we have $\max_{B\in
\mathcal{L}}w(B)$ $-\min_{A\in \mathcal{W}}w(A)$ $=2/3-$ $(1/2+1/(3m))$ $%
=1/6-1/(3m),$ and so, by (\ref{epsilon **}), $\varepsilon ^{\ast }(\widehat{G%
}_{m})\leq $ $1/6-1/(3m),$ producing the equality in (\ref{G^m}).} given any
$\varepsilon <1/6,$ $\widehat{G}_{m}$ is \emph{not} $\varepsilon $-weighted
for all sufficiently large $m.$ Thus, $I$-completeness cannot replace $C$%
-completeness in Theorem \ref{thm:types}.
\end{remark}

\section{Asymptotic tightness of the $\protect\varepsilon ^{\ast }$
convergence rate}

\label{sec asymptotic tightness}

In this section we construct a sequence of games at which the asymptotic
upper estimate $\varepsilon ^{\ast }=O\!\left( 1/n_{\min }\right) $ of the
convergence rate in Theorem \ref{thm:types} is \emph{attained}: the loss-win
excess threshold in that sequence is a constant multiple of $1/n_{\min }$
even as $n_{\min }\rightarrow \infty .$ A defining feature of these games is
their closeness, but not identity, to a weighted game structure with an
integer-valued weight vector. The winning status for coalitions is
determined by their weight with the exception of those that are at the quota
or one level below the quota, for which the winning status indicated by
weight is reversed.\footnote{%
This description puts the games in the category $\mathcal{B}_{2}$ of
Gvozdeva et al. \hyperlink{cite.GHS13}{(2013b)} of what these authors termed
``$\mathbf{2}$-rough games'': a weight vector $W$ determines the coalition's
winning status only outside some \emph{two} consecutive values in the set $%
\left\{ W(A):\emptyset \neq A\subseteq N\right\} $. Our Definition \ref{def
punctured games} imposes further limitations on that weighted structure.}

\subsection{Flip games}

The following definition formalizes the desired properties:

\begin{definition}[\textbf{Flip games}]
\label{def punctured games}A game $(N,\mathcal{W})$ is said to be \emph{%
flipped at the quota }(or, simply, a \emph{flip} game) if there exists a
weight vector $W\in \mathbb{N}^{N}$ and $q\in \mathbb{N}$ with the following
properties:

\begin{itemize}
\item[(i)] for each $i\in N,$ $W_{i}\geq 2;$

\item[(ii)] $\mathcal{W}=\{A:W(A)=q-1\}\cup \{A:W(A)\geq q+1\}$ and,
correspondingly, $\mathcal{L}=\{A:W(A)=q\}\cup \{A:W(A)\leq q-2\}$

\item[(iii)] there \emph{exist} a coalition with weight $q-1$ and a
coalition with weight $q;$

\item[(iv)] there\emph{\ }exists\emph{\ no} \emph{coalition} weighing either
$q-2$ or $q+1.$
\end{itemize}
\end{definition}

Observe that (i) in Definition \ref{def punctured games} ensures that $(N,%
\mathcal{W})$ fully adheres to the definition of a simple game because $%
\mathcal{W}$, as defined in (ii), is up-closed. Indeed, if $A\in \mathcal{W}$
and $i\notin A$ then
\begin{equation*}
W(A\cup \{i\})\overset{\text{(ii)}}{\geq }q-1+W_{i}\overset{\text{(i)}}{\geq
}q+1,
\end{equation*}%
and so $A\cup \{i\}\in \mathcal{W}.$

The weight-gap condition (iv) in Definition \ref{def punctured games} is
somewhat unusual, but flip games exist in abundance, and provide useful
examples. The next proposition shows that $\varepsilon ^{\ast }(G)$ has a
simple expression in some flip games.

\begin{proposition}[$\protect\varepsilon ^{\ast }$ \textbf{for flip games}]
\label{prop gappoed weighted games}If $G=(N,\mathcal{W})$ is a flip game
with weight vector $W\in \mathbb{N}^{n}$, then it is $C$-complete and
\begin{equation}
\varepsilon ^{\ast }(G)\leq 1/W(N).  \label{epslilon*<=w(n)}
\end{equation}%
Furthermore, if there are $a\in \operatorname{conv}\,\overline{\mathcal{W}}$ and $%
b\in \operatorname{conv}\,\overline{\mathcal{L}}$ with $b-a=\bar{n}/W(N)$ then (\ref%
{epslilon*<=w(n)}) holds \emph{with equality.}
\end{proposition}

\subsection{Games with $\protect\varepsilon ^{\ast }=\Theta \left( 1/n_{\min
}\right) $}

We next apply Proposition \ref{prop gappoed weighted games} to the following
sequence of flip games.

\begin{definition}[\textbf{Three-class games} $\left\{ \widetilde{G}%
_{m}\right\} _{m\geq 5}$]
\label{Family G tilda}For each $m\geq 5$, define a $3m$-player game $%
\widetilde{G}_{m}$ as follows. Let its player set $N$ be a union of three
disjoint sets, $N_{1},N_{2},N_{3},$ each of size $m$, and let $W\in \mathbb{N%
}^{N}$ be an integer weight vector with $W_{i}=20$ for every $i\in N_{1},$ $%
W_{i}=15$ for every $i\in N_{2}$ and $W_{i}=12$ for every $i\in N_{3}.$ A
coalition $A\subseteq N$ wins iff
\begin{equation*}
W(A)=59\qquad \text{or}\qquad W(A)\geq 61.
\end{equation*}%
(In principle, $W$ depends on $m$ through the size of $N$, but for
convenience the reference to $m$ is dropped.)
\end{definition}

\begin{proposition}[\textbf{Games with }$\protect\varepsilon ^{\ast }=\Theta
\left( 1/n_{\min }\right) $]
\label{Games with big theta}For each $m\geq 5,$ $\widetilde{G}_{m}$ is a
flip game with $W$ given above and $q=60;$ in particular, by Proposition \ref%
{prop gappoed weighted games}, $\widetilde{G}_{m}$ is $C$-complete. Its
player classes are exactly $N_{1},N_{2},N_{3}$ (and, accordingly, $n_{\min
}=m$), and
\begin{equation*}
0<\varepsilon ^{\ast }(\widetilde{G}_{m})=1/(47m)\text{ }(=1/(47n_{\min })).
\end{equation*}
\end{proposition}

The following is an obvious implication of Proposition \medskip \ref{Games
with big theta}:

\begin{corollary}[\textbf{Asymptotic tightness of the convergence rate in
Theorem \protect\ref{thm:types}}]
There is a sequence of $C$-complete games with three player classes in which
$n_{\min }\rightarrow \infty $ and $0<\varepsilon ^{\ast }=\Theta (1/n_{\min
})$. It follows that the order estimate $O(1/n_{\min })$ of the rate at
which $\varepsilon ^{\ast }$ converges to $0$ (as established in Theorem \ref%
{thm:types}) cannot be improved in general.
\end{corollary}

\begin{remark}[\textbf{Proper games with }$\protect\varepsilon ^{\ast
}=\Theta \left( 1/n_{\min }\right) $]
\label{strong or proper}For any $m\geq 5,$ the game $\widetilde{G}_{m}=(N,%
\mathcal{W})$ is strong but not proper. This feature can, however, be
reversed by looking at the dual game $\widetilde{G}_{m}^{\ast }=(N,\mathcal{W%
}^{\ast })$, where $A\in \mathcal{W}^{\ast }\Longleftrightarrow $ $%
N\diagdown A\notin \mathcal{W}$, and so the winning weight is either at
least $47m-58$ \emph{or} equal to $47m-60.$ Similar conclusions are
obtained: $\widetilde{G}_{m}^{\ast }$ is a flip game for the weight vector $W
$ described in Definition \ref{Family G tilda} and $q=47m-59,$ it is $C$%
-complete and \emph{proper} (though not strong), with%
\begin{equation*}
0<\varepsilon ^{\ast }(\widetilde{G}_{m}^{\ast })\overset{\text{(*)}}{=}%
\varepsilon ^{\ast }(\widetilde{G}_{m})=1/(47m)=\Theta \left( 1/n_{\min
}\right) .
\end{equation*}%
The equality (*) holds because for any $w\in \Delta (N)$, we have
\begin{equation*}
\max_{B\in \mathcal{L}^{\ast }}w(B)-\min_{A\in \mathcal{W}^{\ast
}}w(A)=\max_{A\in \mathcal{W}}w(N\setminus A)-\min_{B\in \mathcal{L}%
}w(N\setminus B)=\max_{B\in \mathcal{L}}w(B)-\min_{A\in \mathcal{W}}w(A),
\end{equation*}%
which then implies (*) via (\ref{epsilon **}).
\end{remark}

\begin{remark}[S\textbf{maller games with }$\protect\varepsilon ^{\ast }>0$]

\label{9-player epsilon*>0} As explained in the discussion following Theorem %
\ref{thm:main} in Section \ref{sec computing epsilon*}, the importance of
the games $\left\{ \widetilde{G}_{m}\right\} _{m\geq 5}$ lies also in the
fact that they are simultaneously $C$-complete and \emph{not} roughly
weighted (i.e., $\varepsilon ^{\ast }(\widetilde{G}_{m})>0$). That is, the
sequence $\left\{ \widetilde{G}_{m}\right\} _{m\geq 5}$ demonstrates that $C$%
-completeness does not imply rough weightedness no matter how large the
player set is. However, the smallest of these games, $\widetilde{G}_{5},$
has $15$ players, and one may ask whether there are any \emph{smaller }$C$%
-complete games that are not roughly weighted. The answer is in the
affirmative. It can be verified that a $10$-player flip game with a weight
vector $W=\left( 64,60,54,51,46,43,40,36,31,24\right) $ and $q=225$ --- that
is constant-sum and $C$-complete --- has $\varepsilon ^{\ast
}=1/W(N)=1/449>0.$\footnote{%
Indeed, consider the following coalitions, in which players are identified
with their respective weights: $A_{1}=$ $\{54,51,43,40,36\},$ $A_{2}=$ $%
\{64,46,43,40,31\},$ $A_{3}=$ $\{60,51,46,43,24\},$ $A_{4}=$ $%
\{60,54,46,40,24\}$ and $A_{5}=$ $\{64,54,51,31,24\},$ all having the weight
of $224=q-1$ (and therefore in $\mathcal{W}$), and $B_{1}=$ $%
\{54,51,46,43,31\},$ $B_{2}=$ $\{60,51,43,40,31\},$ $B_{3}=$ $%
\{64,51,46,40,24\},$ $B_{4}=$ $\{64,54,43,40,24\},$ $B_{5}=$ $%
\{64,60,46,31,24\}$ and $B_{6}=$ $\{51,43,40,36,31,24\},$ all having the
weight of $225=q$ (and therefore in $\mathcal{L}$). No two of the ten
players are equally desirable, so all classes are singletons, $\bar{n}%
=(1,\ldots ,1)$, and profiles are indicator vectors. Denote by $a_{i}\in
\overline{\mathcal{W}}$ and $b_{j}\in \,\overline{\mathcal{L}}$ the
indicator vector of each $A_{i},$ $B_{j}$ above, and let $a=$ $\left(
9a_{1}+148a_{2}+98a_{3}+60a_{4}+134a_{5}\right) /449$ $\in \mathrm{conv}\,%
\overline{\mathcal{W}}$ and $b=$ $\left(
114b_{1}+42b_{2}+76b_{3}+90b_{4}+117b_{5}+10b_{6}\right) /449\in \mathrm{conv%
}\,\overline{\mathcal{L}}.$ Then $b-a=$ $(1,\ldots ,1)/449$ $=\bar{n}/W(N)$,
and so $\varepsilon ^{\ast }=1/449$ by Proposition \ref{prop gappoed
weighted games}.} Also, all of its ten $9$-player reduced games --- obtained
by removing a single player and subtracting his weight from $q$ --- are $C$%
-complete, strong, and have $\varepsilon ^{\ast }>0.$ (Since all $C$%
-complete games with at most $8$ players are weighted, these reduced games
are of the \emph{minimal} size at which non-weightedness of a $C$-complete
game can be obtained.)
\end{remark}

\subsection{Acyclicity of $\succ _{C}$ does not imply (any of our notions
of) weightedness}

\label{sec acyclicity}

The game sequence $\left\{ \widetilde{G}_{m}\right\} _{m\geq 5}$ from
Section \ref{sec asymptotic tightness} delivers another interesting fact. It
is easy to see that the strict part $\succ _{C}$ of the $C$-desirability
relation $\succeq _{C}$ in a weighted game is necessarily \emph{acyclic},
and there was a conjecture that acyclicity of $\succ _{C}$ may also play a
converse role, of being sufficient for weightedness in a $C$-complete game.
Einy \hyperlink{cite.Einy85}{(1985)} refuted this conjecture by exhibiting a
non-weighted $C$-complete proper game with an acyclic $\succ _{C}.$ However,
as mentioned in Section \ref{sec computing epsilon*}, that game is \emph{%
roughly weighted}, with $\varepsilon ^{\ast }=0$.

It turns out that acyclicity of $\succ _{C}$ in a $C$-complete game is
actually even less potent. The acyclicity does not suffice to guarantee
\emph{even rough weightedness}, or a weaker version ($\varepsilon $%
-weightedness for a small $\varepsilon >0$). Having a large number of
(non-null) players in the game also does not change that, except in the
asymptotic sense of Theorem \ref{thm:types}. The game sequence $\left\{
\widetilde{G}_{m}\right\} _{m\geq 5}$ given in Definition \ref{Family G
tilda} attests precisely to this lack of generalized weightedness despite
the acyclicity of $\succ _{C}$ and the rising number of players:

\begin{proposition}[\textbf{Acyclicity of }$\succ _{C}$\textbf{\ }$%
\nRightarrow $\textbf{\ rough or small-}$\protect\varepsilon $\textbf{%
-weightedness}]
\label{acyclicity and non-weightedness}For every $m\geq 5$, the strict $C$%
-desirability relation $\succ _{C}$ of $\widetilde{G}_{m}$ is acyclic
(whereas $\widetilde{G}_{m}$ is not roughly weighted by Proposition \ref%
{Games with big theta}).\footnote{%
In the terminology of Taylor and Zwicker \hyperlink{cite.TZ96}{(1996},
Theorem 5.1), $\widetilde{G}_{m}$ carries an\emph{\ SSAP-weighting} without
being roughly weighted, and so Proposition \ref{acyclicity and
non-weightedness} shows that these two notions of weighting are not
identical inside the domain of $C$-complete games.} As a particular example,
$\widetilde{G}_{5}$ is a $15$-player $C$-complete game with an acyclic $%
\succ _{C}$ that is not roughly-weighted, and, more generally, not $%
\varepsilon $-weighted for any $\varepsilon <$ $\varepsilon ^{\ast }(%
\widetilde{G}_{5})=$ $1/(47\times 5)$ $=1/235$.
\end{proposition}

\begin{remark}[\textbf{A \emph{constant-sum} game with acyclic }$\succ _{C}$%
\textbf{\ and }$\protect\varepsilon ^{\ast }>0$]
\label{rem constant-sum acyclic}Taylor and Zwicker \hyperlink{cite.TZ95}{%
(1995}, Section 5) demonstrated that the acyclicity of $\succ _{C}$ in a $C$%
-complete game need not imply weightedness even in \emph{constant-sum}
games, but their smallest counterexample --- a constant sum extension of a
slivered $3\times 3$ magic square-game --- has $28$ players (and is roughly
weighted). Our sequence $\left\{ \widetilde{G}_{m}\right\} _{m\geq 5}$ is
convenient for showing that acyclicity of $\succ _{C}$ in a $C$-complete
game does not imply even rough weightedness, for any number of players, but
these games are not constant-sum.\footnote{%
Also, and critically, the constant-sum extension of any $\widetilde{G}_{m}$
has a $\succ _{C}$-cycle. That is because, by Taylor and Zwicker \hyperlink{cite.TZ95}%
{(1995}, Theorem 4.4(4)), $\succ _{C}$-acyclicity of the extension would
imply $\succ _{C}$-acyclicity of both $\widetilde{G}_{m}$ and its dual $%
\widetilde{G}_{m}^{\ast };$ however, $\widetilde{G}_{m}^{\ast }$ has a $%
\succ _{C}$-cycle. By choosing distinct $a_{1},a_{2}\in $ $N_{1}$, $%
b_{1},b_{2}$ $\in N_{2}$ and $d_{1},d_{2},d_{3}$ $\in N_{3},$ the following
is a $\succ _{C}$-cycle in $\widetilde{G}_{m}^{\ast }$: $\{a_{1},a_{2},b_{1}%
\}\ $\ $\succ _{C}\ \{a_{1},d_{1},d_{2},d_{3}\}\ $\ $\succ _{C}\ $\ $%
\{b_{1},b_{2},d_{1},d_{2}\}\ $\ $\succ _{C}\ \{a_{1},a_{2},b_{1}\}.$}
However, there is a smaller and constant-sum alternative to $\widetilde{G}%
_{m}$. Indeed, a $14$-player flip game with four player classes of sizes $%
\overline{n}=\left( 3,3,4,4\right) $, weights $51,30,13,3$ per player in the
respective classes, and $q=154,$ is constant-sum and $C$-complete with an
acyclic $\succ _{C}$ and $\varepsilon ^{\ast }=1/307>0.$ (The proof is
similar to that of Propositions \ref{Games with big theta} and \ref%
{acyclicity and non-weightedness}, and is sketched in Section \ref{sec
sketch} of the Appendix)$.$
\end{remark}

\section{Asymptotic $\protect\alpha $-rough weightedness}

\label{alpha-weightedness section}

The asymptotic result of Theorem \ref{thm:types} largely translates into the
alternative framework of $\alpha $-roughly weighted games of Freixas and
Kurz \hyperlink{cite.FK14}{(2014)}: diminishing values of the non-negative
part $\max \{\varepsilon ^{\ast },0\}$ of loss-win-excess threshold are
precisely equivalent to the fall of Freixas and Kurz's critical threshold $%
\mu $ to 1 (corresponding to rough weightedness) for \emph{proper} $C$%
-complete games, and to a large extent elsewhere. In what follows, we link
their framework to ours, and state the results.

Denote by $W\in \mathbb{R}_{+}^{n}$ any weight vector, keeping the lowercase
$w$ for normalized weight vectors in $\Delta \left( N\right) .$

\begin{definition}[$\protect\alpha $-\textbf{rough weightedness and
threshold }$\protect\mu $\textbf{; }\emph{Freixas and Kurz \hyperlink{cite.FK14}%
{(2014)}}]
\label{alpha definition}Given $\alpha \geq 1,$ a game $G=(N,\mathcal{W})$ is
$\alpha $-\emph{roughly weighted} if for some weight vector $W\in \mathbb{R}%
_{+}^{n}$ the following holds:%
\begin{equation}
A\in \mathcal{W}\ \Longrightarrow \ W(A)\ \geq 1,\qquad B\in \mathcal{L}\
\Longrightarrow \ W(B)\ \leq \alpha .  \label{alpha inequalities}
\end{equation}%
The \emph{critical threshold} $\mu (G)\geq 1$ is the least such $\alpha $.
\end{definition}

In effect, as stated in Hof et al. \hyperlink{cite.hot}{(2020)}, $\mu (G)$
is representable as the least value of the \emph{ratio} between the maximal
losing weight and the minimal winning weight:%
\begin{equation}
\mu (G)=\max \left\{ 1,\min_{W\in \mathbb{R}_{+}^{n}}\max_{A\in \mathcal{W}%
,B\in \mathcal{L}}\left( W(B)/W(A)\right) \right\}   \label{mu(G) original}
\end{equation}%
(with the ratio defined as $\infty $ whenever $W(A)=0$). In contrast, the
threshold $\varepsilon ^{\ast }\left( G\right) $ is the least value of the
\emph{difference} between $\max_{B\in \mathcal{L}}w(B)$ and $\min_{A\in
\mathcal{W}}w(A)$ (for a normalized $w\in \Delta (N)$). Thus, \emph{in
general}, the optimal weight vectors may be different (even after rescaling)
for $\mu (G)$ and $\varepsilon ^{\ast }\left( G\right) .$ However, for
proper games, convergence to rough weightedness is fully equivalent under
both regimes, and, for constant-sum games, the loss-win-excess and the
critical threshold are moreover functionally tied:

\begin{proposition}[\textbf{Relation between }$\protect\mu (G)$ \textbf{and}
$\protect\varepsilon ^{\ast }\left( G\right) $]
\label{relation between thresholds}For any game $G=(N,\mathcal{W}),$
\begin{equation}
1+\varepsilon ^{\ast }\left( G\right) \leq \ \mu (G).
\label{comparing epsilon * and mu1}
\end{equation}%
If $G$ is \emph{proper} and $\varepsilon ^{\ast }(G)\in \lbrack 0,1),$ then
\begin{equation}
\mu (G)\leq \left( 1+\varepsilon ^{\ast }(G)\right) /\left( 1-\varepsilon
^{\ast }(G)\right) .  \label{comparing epsilon * and mu}
\end{equation}%
If $G$ is also strong (and therefore \emph{constant-sum}) then (\ref%
{comparing epsilon * and mu}) holds with\emph{\ }equality.
\end{proposition}

In light of Proposition \ref{relation between thresholds}, Theorem \ref%
{thm:types} can be restated for proper games in terms of the limiting
behavior of the critical threshold $\mu (G),$ as follows.

\begin{corollary}[\textbf{Asymptotics expressed by }$\protect\mu $\textbf{\
for proper games}]
\label{asymptotics for critical threshold}Let $\left\{ G_{n}\right\} _{n\geq
1}$ be a sequence of $n$-player $C$-complete and \emph{proper} games with $%
\lim_{n\rightarrow \infty }n_{\min }$ $=\infty $ and $d_{n}=$ $\exp
(o(n_{\min })).$ Then%
\begin{equation*}
\lim_{n\rightarrow \infty }\mu \left( G_{n}\right) =1.
\end{equation*}%
Furthermore, if $d_{n}$ is bounded along the sequence then
\begin{equation*}
\mu \left( G_{n}\right) =1+\;O\!\left( 1/n_{\min }\right) .
\end{equation*}
\end{corollary}

When a $C$-complete game $G$ is \emph{not proper}, the asymptotic
weightedness of Theorem \ref{thm:types} may, in principle, fail to have the
corresponding expression in terms of $\alpha $-rough weightedness. Although (%
\ref{comparing epsilon * and mu1}) holds for all (not necessarily proper)
games $G$, and so $\mu (G)\rightarrow 1$ always implies $\max \{\varepsilon
^{\ast }(G),0\}\rightarrow 0$, the converse may be false. A case in point is
the sequence $\left\{ \widetilde{G}_{m}\right\} _{m\geq 5}$ in Definition %
\ref{Family G tilda}, in which the critical threshold $\mu (\widetilde{G}%
_{m})$ stays bounded away from $1$ even though $\varepsilon ^{\ast }\left(
\widetilde{G}_{m}\right) \rightarrow _{m\rightarrow \infty }0$:

\begin{proposition}[\textbf{Asymptotics may be missed by }$\protect\mu $%
\textbf{\ for improper games}]
\label{divergence of mu and epsilon*}For each $m\geq 5,$ $\mu (\widetilde{G}%
_{m})=60/59>1,$ whereas (by Proposition \ref{Games with big theta}) $%
\lim_{m\rightarrow \infty }\varepsilon ^{\ast }(\widetilde{G}_{m})=0.$
\end{proposition}

Having pointed out that the limiting behavior of $\varepsilon ^{\ast }$ and $%
\mu $ may decouple, we should note that this is, in fact, a rather isolated
phenomenon. Specifically for $\widetilde{G}_{m}=(N,\mathcal{W})$ used in
Proposition \ref{divergence of mu and epsilon*}, the non-convergence of $\mu
$ is an artifact of $\min_{A\in \mathcal{W}}w^{\ast }(A)$ tending to zero
for any $w^{\ast }$ attaining the minimum in the expression (\ref{epsilon **}%
) for $\varepsilon ^{\ast }$; it is for this reason that $\max_{A\in
\mathcal{W},B\in \mathcal{L}}\frac{W(B)}{W(A)}$ in (\ref{mu(G) original})
happens to be bounded away from $1$ for all $m$. In words, $\widetilde{G}_{m}
$ becomes ``very'' improper when $m\rightarrow \infty :$ all coalitions
containing at least $6$ players are automatically winning regardless of the
total number of players in the game. Thus, with the optimal weighting being
normalized to $\Delta (N),$ the weight of all coalitions that matter in
determining $\mu (\widetilde{G}_{m})$ and $\varepsilon ^{\ast }\left(
\widetilde{G}_{m}\right) $ approaches zero. Outside this case, asymptotic
weightedness is fully reflected by the behavior of the critical threshold $%
\mu $:

\begin{proposition}[\textbf{When }$\protect\mu $\textbf{\ captures
asymptotic weightedness}]
\label{Corollary non proper} Let $\left\{ G_{n}\right\} _{n\geq 1}$ be a
sequence of $n$-player $C$-complete games with $\lim_{n\rightarrow \infty
}n_{\min }$ $=\infty $ and $d_{n}=$ $\exp (o(n_{\min })).$ For each $%
G_{n}=(N,\mathcal{W})$, pick $w^{\ast }$ attaining the minimum in the
expression (\ref{epsilon **}) for $\varepsilon ^{\ast }\left( G_{n}\right) $%
, and let $q_{n}^{\ast }=$ $\min_{A\in \mathcal{W}}w^{\ast }(A)$. If $\lim
\inf_{n\rightarrow \infty }q_{n}^{\ast }>0$ then%
\begin{equation*}
\lim_{n\rightarrow \infty }\mu \left( G_{n}\right) =1,
\end{equation*}%
and, if $d_{n}$ is uniformly bounded along the sequence, then%
\begin{equation*}
\mu \left( G_{n}\right) =1+\;O\!\left( 1/n_{\min }\right) .
\end{equation*}
\end{proposition}

\section{Approximate weightedness and vanishing misclassification}

\label{sec misclassification}

\subsection{General misclassification measure}

Given a game $G=(N,\mathcal{W}),$ under an $\varepsilon ^{\ast }=\varepsilon
^{\ast }(G)$-optimal weight system (consisting of a weight vector $w^{\ast
}\in \Delta (N)$ attaining the minimum in (\ref{epsilon **}) and a quota$\
q^{\ast }=$ $\min_{A\in \mathcal{W}}w^{\ast }(A)$), any coalition $A$ whose
winning status is mislabeled by the weights is within $\varepsilon ^{\ast }$
in terms of its weight to the right status determination, but that fact
itself does not in any way remedy the mislabeling of $A.$ We show here that
the smallness of $\varepsilon ^{\ast }$ is matched by the number of
mislabeled coalitions in the central case: when the number of classes is
bounded, the least achievable \emph{fraction of coalitions} misclassified by
a weight system tends to zero when $n_{\min }\rightarrow \infty $, just as
does the loss-win-excess threshold $\varepsilon ^{\ast }$ by Theorem \ref%
{thm:types}.

\begin{definition}[\textbf{The misclassification measure}]
\label{def misclassification}For $G=(N,\mathcal{W})$, let
\begin{equation}
\delta (G)=\ \left( 1/2^{n}\right) \min_{W\in \mathbb{R}_{+}^{N},\,q\in
\mathbb{R}}\left\vert \left\{ A\subseteq N:\ 
\begin{array}{c}
A\in \mathcal{W}\text{ and }W(A)<q, \\
\text{or }A\in \mathcal{L}\text{ and }W(A)\geq q%
\end{array}%
\right\} \right\vert   \label{deltag}
\end{equation}%
be the least fraction of coalitions mislabeled by a weight system, and call $%
\delta (G)$ the \emph{misclassification measure }of $G$.
\end{definition}

Equivalently, the misclassification measure $\delta (G)$ is also the least
probability with which a weight system misclassifies a random coalition $A$
chosen w.r.t. the uniform distribution on $2^{N}.$

\begin{remark}[\textbf{The outlier value}]
Our misclassification measure has an analogue in the recent literature,
which is the \emph{outlier value} of Albareda-Sambola et al.\ \hyperlink{cite.amr}%
{(2025}, Definition 4). Denoted by $\Psi _{M}(G)$ for a given $M>0,$ it is
defined as the minimum of $\sum_{A\subseteq N}\psi _{A}$ subject to
\begin{equation*}
\begin{array}{ll}
W(A)\geq q-M\psi _{A} & (\forall A\in \mathcal{W}), \\
W(B)\leq q-1+M\psi _{B} & (\forall B\in \mathcal{L}), \\
W_{i}\geq 0\ (\forall i\in N),\ q\geq 1, & \psi _{A}\in \{0,1\}\ (\forall
A\subseteq N).%
\end{array}%
\end{equation*}%
Thus, $\Psi _{M}$ counts \emph{outliers}, or coalitions $A$ with $\psi
_{A}=1 $ whose constraint is relaxed by $M$ units in order to make the rest
of the game weighted.

Our misclassification measure $\delta (G)$ is related to the \emph{normalized%
} outlier value $\Psi _{M}(G)/2^{n},$ i.e., the minimal fraction of
coalitions that must be declared outliers in order to correctly weight the
remaining coalitions. It is easy to see that for any game $G$, $\delta
(G)=\Psi _{M}(G)/2^{n}$ whenever $M$ is sufficiently large. Indeed, on the
one hand, any feasible solution for the outlier value problem counts the
mislabeled coalitions, and so $\delta (G)\leq \Psi _{M}(G)/2^{n}.$ On the
other hand, any weight system $W,q$ attaining the minimum in (\ref{deltag})
must have $q>0,$\footnote{%
If $q\leq 0$, every losing coalition is misclassified. But a weight system
that gives every player weight $1$, with $q=1/2,$ misclassifies only
non-empty losing coalitions, and so $q\leq 0$ is never optimal.} and so the
finitely many correctly classified losing coalitions have a positive minimum
slack, which can be scaled to at least one. That would make the weight
system feasible for the outlier value problem, for all sufficiently large $%
M, $ yielding $\delta (G)\geq \Psi _{M}(G)/2^{n}.$
\end{remark}

\begin{proposition}[\textbf{Upper bound on} $\protect\delta $]
\label{prop:misclass}For any $C$-complete game $G$,
\begin{equation}
\delta (G)\leq \left( \sqrt{2d/\pi }(\log _{2}d+2)+2C_{0}\right) /\sqrt{%
n_{\min }},  \label{delta bounded by epsilon*}
\end{equation}%
where $C_{0}\leq 0.56$ is the Berry--Esseen constant.
\end{proposition}

Given (\ref{delta bounded by epsilon*}), it is an immediate corollary that $%
\delta $ becomes vanishingly small whenever the growth in the number of
player classes is sufficiently slow compared to $n_{\min }$:

\begin{corollary}[\textbf{Asymptotic smallness of }$\protect\delta $]
\label{cor when delta g is small}Let $\left\{ G_{n}\right\} _{n\geq 1}$ be a
sequence of $n$-player $C$-complete games in which $\lim_{n\rightarrow
\infty }n_{\min }=\infty $ and the number of player classes $d_{n}$
satisfies $(\log _{2}d_{n}+2)\sqrt{d_{n}}=o(\sqrt{n_{\min }})$. Then%
\begin{equation}
\lim_{n\rightarrow \infty }\delta \left( G_{n}\right) =0.
\label{deltaG_n goes to 0}
\end{equation}%
If the number of player classes $d_{n}$ is \emph{bounded}, then%
\begin{equation*}
\delta (G_{n})=O(1/\sqrt{n_{\min }}).
\end{equation*}
\end{corollary}

Accordingly, with a bounded number of player classes, $\varepsilon ^{\ast
}=O(1/n_{\min })$ and $\delta =O(1/\sqrt{n_{\min }})$ by Theorem \ref%
{thm:types} and Corollary \ref{cor when delta g is small}, but $%
\lim_{n\rightarrow \infty }\delta \left( G_{n}\right) =$ $\lim_{n\rightarrow
\infty }\varepsilon ^{\ast }\left( G_{n}\right) =0$ is also obtained when
the growth in the number of classes is sufficiently slow. Thus, the
``separation'' and ``misclassification'' based measures of the proximity to
weightedness are well aligned in those central cases.

\subsection{Misclassification with $\protect\varepsilon ^{\ast }$-attaining
weights}

It is important to note that Definition \ref{def misclassification} does not
in any way restrict the set of weight and quota combinations $W\in \mathbb{R}%
_{+}^{N},\,q\in \mathbb{R}$ over which the minimal mislabeled fraction is
taken in (\ref{deltag}). In particular, the weight vector $W^{\ast }$ at
which $\delta (G)$ in (\ref{deltag}) is attained may well be different (even
after normalization) from $w^{\ast }\in \Delta (N)$ attaining $\varepsilon
^{\ast }(G)$ in (\ref{epsilon **}). That is, if%
\begin{equation*}
\Delta ^{\ast }(G)=\arg \min_{w\in \Delta (N)}\left( \max_{B\in \mathcal{L}%
}w(B)-\min_{A\in \mathcal{W}}w(A)\right) ,
\end{equation*}%
and%
\begin{equation*}
\delta ^{\ast }(G)=\left( 1/2^{n}\right) \min_{w^{\ast }\in \Delta ^{\ast
}(G),\,q\in \mathbb{R}}\left\vert \left\{ A\subseteq N:\ 
\begin{array}{c}
A\in \mathcal{W}\text{ and }w^{\ast }(A)<q, \\
\text{or }A\in \mathcal{L}\text{ and }w^{\ast }(A)\geq q%
\end{array}%
\right\} \right\vert
\end{equation*}%
is the minimal mislabeled fraction under $w^{\ast }\in \Delta ^{\ast }(G)$
--- a weight vector producing the least weight difference between the losing
and the winning coalitions --- then $\delta ^{\ast }(G)$ may be larger than $%
\delta (G).$ Because we view $\varepsilon ^{\ast }$ as the measure of
proximity to weightedness, it would be comforting to know that, at least in
the central case, $\varepsilon ^{\ast }\rightarrow 0$ is accompanied not
only by $\delta \rightarrow 0$ (namely, vanishing minimal attainable
fraction mislabeled by weights) but by a stronger feature $\delta ^{\ast
}\rightarrow 0$ (namely, vanishing fraction that is mislabeled specifically
by $\varepsilon ^{\ast }$-attaining weights).

Conveniently, there is a version of (\ref{delta bounded by epsilon*}) that
provides a (weaker) bound on $\delta ^{\ast }(G):$

\begin{proposition}[\textbf{Upper bound on} $\protect\delta ^{\ast }$]
\label{misclassification for epsilon*}For any $C$-complete $n$-player game $%
G,$%
\begin{equation}
\delta ^{\ast }(G)\leq \sqrt{2/\pi }\,(\log _{2}d+2)\frac{\sqrt{n}}{n_{\min }%
}+\frac{2C_{0}}{\sqrt{n_{\min }}}.  \label{upper bound delta(G)}
\end{equation}
\end{proposition}

Thus, for a sequence $\left\{ G_{n}\right\} _{n\geq 1}$ of $n$-player $C$%
-complete games in which $\lim_{n\rightarrow \infty }n_{\min }=\infty $ and $%
d_{n}$ satisfies $(\log _{2}d_{n}+2)\sqrt{n}$ $=o(n_{\min })$, Proposition %
\ref{misclassification for epsilon*} implies that
\begin{equation*}
\lim_{n\rightarrow \infty }\delta ^{\ast }\left( G_{n}\right) =0.
\end{equation*}%
In particular, if $n=O(n_{\min })$ (and, by implication, $d_{n}$ is \emph{%
bounded}) then%
\begin{equation*}
\delta ^{\ast }(G_{n})=O(1/\sqrt{n_{\min }}).
\end{equation*}%
The latter is precisely the case when each player class contains at least
some positive fixed fraction of the total electorate. Therefore, in this
central scenario, we have both $\varepsilon ^{\ast }=O(1/n_{\min })$ (by
Theorem \ref{thm:types}) and $\delta ^{\ast }=O(1/\sqrt{n_{\min }}).$

\section{Concluding remarks}

\label{concluding remarks}

$C$-completeness is by itself too weak to force weightedness of a finite
simple game; what has been shown here is that it forces \emph{approximate}
weightedness when the minimal size $n_{\min }$ of a class of equally
desirable players is large. The loss--win-excess threshold $\varepsilon
^{\ast }$ of any $C$-complete game with $d$ classes --- such that the game
is $\varepsilon $-weighted iff $\varepsilon \geq \varepsilon ^{\ast }$ ---
is shown to have the upper bound of $\left( \log _{2}d+2\right) /n_{\min },$
which involves neither the number of players nor the size of the largest
class. Nothing other than the completeness of $\succeq _{C}$ is needed to
establish that bound; in particular, the strict part $\succ _{C}$ is not
required to be acyclic, which is a necessary feature in fully weighted games.

The threshold $\varepsilon ^{\ast }$ tends to $0$ along any sequence with $%
n_{\min }\rightarrow \infty $ in which the number of classes grows
subexponentially in $n_{\min }$. Also, $\varepsilon ^{\ast }=O(1/n_{\min })$
--- and the bound is tight --- when the number of classes is bounded, e.g.,
when each class holds at least some fixed fraction of the electorate. The
misclassification measure $\delta $, the minimal fraction of coalitions
mislabeled by a weighted majority representation, reflects the fall of $%
\varepsilon ^{\ast }$ to $0$ --- when $n_{\min }\rightarrow \infty $ and the
number of classes is bounded --- by being $O(1/\sqrt{n_{\min }});$ when each
class holds at least a fixed fraction of the electorate, the coarser $\delta
^{\ast }$ that minimizes the fraction that is mislabeled specifically by $%
\varepsilon ^{\ast }$-attaining weights is also $O(1/\sqrt{n_{\min }}).$

We conclude with some questions that are left open and seem worthy of
further investigation.

\subsection{Is full weightedness ever obtainable in the limit?}

Theorem~\ref{thm:types} delivers $\lim_{n_{\min }\rightarrow \infty
}\varepsilon ^{\ast }=0$, but that is asymptotic \emph{rough} weightedness.
The theorem says nothing about the sign of $\varepsilon ^{\ast }$ as it
approaches the limit, and so $\varepsilon ^{\ast }$ may be positive along an
entire sequence of games. The latter is the case with the games $\left\{
\widetilde{G}_{m}\right\} _{m\geq 5}$ in Section \ref{sec asymptotic
tightness}, where under any possible weight system there would be at least
\emph{some} coalitions whose weight attributes to them a wrong winning
status. Is there a natural strengthening of the large-class assumption that
would ensure that $\varepsilon ^{\ast }$ permanently drops below $0$
starting from some player-class size, so that the games in the sequence
eventually become weighted in the usual strict sense?

\subsection{Does $\protect\varepsilon ^{\ast }$ converge to zero faster in
constant-sum games?}

For the sequence $\left\{ \widetilde{G}_{m}\right\} _{m\geq 5}$ of flip
games in Section \ref{sec asymptotic tightness}, $\varepsilon ^{\ast }$ was
found to be proportional to $1/m,$ thereby showing that the asymptotic
estimate $\varepsilon ^{\ast }=O(1/n_{\min })$ cannot be improved in
general. However, no similar claim can be made for a sequence of \emph{%
constant-sum} flip games. Indeed, it can be shown that $\max \{\varepsilon
^{\ast },0\}=$ $O(1/n_{\min }^{2})$ for such games:

\begin{proposition}[\textbf{Bound on }$\protect\varepsilon ^{\ast }$ \textbf{%
for constant-sum flip games}]
\label{O(1/n^2)}Assume that $G$ is a \emph{flip game} with a weight vector $%
W\in \mathbb{N}^{N}$ that is constant on each player class, and $q=\left(
W(N)+1\right) /2\in \mathbb{N}$. Then $G$ is constant-sum, $C$-complete,
and, irrespective of $d,$
\begin{equation*}
\varepsilon ^{\ast }(G)\leq 6/n_{\min }^{2}.
\end{equation*}
\end{proposition}

We do not know whether an upper bound that is \emph{quadratic} in $1/n_{\min
}$ applies asymptotically to $\varepsilon ^{\ast }$ in general $C$-complete
constant-sum games. However, its appearance in flip games may be relevant
for the general constant-sum case, just as the $\Theta (1/n_{\min })$%
-convergence of $\varepsilon ^{\ast }$ in some non-constant-sum flip games
supports the asymptotic tightness of the estimate $\varepsilon ^{\ast
}=O(1/n_{\min })$ for all games.

\subsection{Does $\protect\delta $ (or $\protect\delta ^{\ast }$) converge
to zero faster than $O(1/\protect\sqrt{n_{\min }})$?}

Our misclassification measures are combinatorial, being fractions mislabeled
by a weighted approximation, and thus inherently distinct from a geometric
separation measure embodied by the loss-win threshold $\varepsilon ^{\ast }.$
Marked differences in the limit behavior of the two measure types may
therefore be expected. Our results show that $\varepsilon ^{\ast }$ and $%
\delta $ tend to converge to $0$ in tandem, but while $\varepsilon ^{\ast
}=O(1/n_{\min })$ is a tight bound, nothing indicates that tightness also
applies to $\delta =O(1/\sqrt{n_{\min }})$. Indeed, the sequence $\left\{
\widetilde{G}_{m}\right\} _{m\geq 5}$ that supports the tightness of $%
\varepsilon ^{\ast }=O(1/n_{\min })$ has a distinct effect on $\delta ^{\ast
}$: because all coalitions with six or more players are winning, and because
in computing the upper bound for $\delta $ or $\delta ^{\ast }$ one may
assume that mislabeling only affects the losing coalitions,\footnote{%
See the definition of $q^{\ast }$ in the proofs of Propositions \ref%
{prop:misclass} and \ref{misclassification for epsilon*} (Sections \ref%
{proof of prop on misclassification} and \ref{sec proof of equality}).} $%
\delta ^{\ast }(\widetilde{G}_{m})$ is bounded by the fraction of coalitions
with at most $5$ players, which is $o\left( 2^{-n_{\min }}\right) =o\left(
2^{-m}\right) .$ The indication from within the proof of Proposition \ref%
{prop:misclass} (Section \ref{proof of prop on misclassification}) also
supports the hunch that $\delta $ could be significantly below the bound on
many occasions.\footnote{%
The probability of the random weight $W(\mathbf{A})$ being in an interval $I$
that Lemma \ref{lem:be} estimates corresponds to the worst-case scenario; if
the entire $I$ is sufficiently far from the mean $W(N)/2,$ then that
probability estimate is significantly lower by Hoeffding's inequality.}

\section{Appendix}

\label{Appendix A}

\subsection{Proof of Proposition \protect\ref{lem:inc}}

Assume first that $(N,\mathcal{W})$ is not $C$-complete, and so some $A,B$
are $\succeq _{C}$-incomparable. Then there exist $R,R^{\prime }\subseteq
N\setminus (A\cup B)$ for which $A\cup R\in \mathcal{W}$, $B\cup R\in
\mathcal{L}$, $B\cup R^{\prime }\in \mathcal{W}$, $A\cup R^{\prime }\in
\mathcal{L}$. Because $R,R^{\prime }$ are disjoint from $A,B,$ we have $\chi
_{A\cup R}+\chi _{B\cup R^{\prime }}$ $=\chi _{A}+\chi _{B}+\chi _{R}+\chi
_{R^{\prime }}$ $=\chi _{B\cup R}+\chi _{A\cup R^{\prime }}$, exhibiting a $%
2 $-trade.

In the opposite direction, assume that $(N,\mathcal{W})$ is not $2$-trade
robust, and thus there exist $X_{1},X_{2}\in \mathcal{W}$ and $%
Y_{1},Y_{2}\in \mathcal{L}$ such that
\begin{equation*}
\chi _{X_{1}}+\chi _{X_{2}}\;=\;\chi _{Y_{1}}+\chi _{Y_{2}}.
\end{equation*}%
This equality means precisely that $X_{1}\cap X_{2}=Y_{1}\cap Y_{2}$ and
that $X_{1}\cup X_{2}=Y_{1}\cup Y_{2}$. Let $A=X_{1}\setminus Y_{1}$, $%
B=Y_{1}\setminus X_{1},$ and $R=$ $X_{1}\cap Y_{1}$, $R^{\prime }=$ $%
X_{2}\cap Y_{2}.$ Then $R,R^{\prime }$ are disjoint from $A\cup B$, and $\
A\cup R=X_{1}\in \mathcal{W}$, $B\cup R=Y_{1}\in \mathcal{L}$, yet $A\cup
R^{\prime }=Y_{2}\in \mathcal{L}$ and $B\cup R^{\prime }=X_{2}\in \mathcal{W}%
.$ Consequently, $A$ and $B$ are $\succeq _{C}$-incomparable, and so the
game is not $C$-complete$.$

\subsection{Proof of Proposition \protect\ref{lem:orbits}}

The relation $\succeq _{I}$ is defined based on the set $\mathcal{W}$ alone,
and so for any $\sigma \in \Gamma (G)$ the following holds: $i\succeq _{I}j$
if and only if $\sigma (i)\succeq _{I}\sigma (j)$. Therefore $\sigma $ moves
equivalence classes of $\sim _{I}$ into equivalence classes and preserves
their $\succ $-ranking, which is only possible if $\sigma (N_{t})=N_{t}$ for
every $t$. Conversely, it follows from definition of $\sim _{I}$ that
transposition of any $i$ and $j$ inside a class $N_{t}$ is an automorphism.
Since any permutation that preserves each $N_{t}$ is a product of such
transpositions, it belongs to $\Gamma (G).$

\subsection{Proof of Proposition \protect\ref{prop:sumset}}

If $G=(N,\mathcal{W})$ is not $C$-complete, then there are $A,B\subseteq N$
and $R,R^{\prime }\subseteq N\diagdown (A\cup B)$ with $A\cup R\in \mathcal{W%
},$ $B\cup R\in \mathcal{L},$ but $B\cup R^{\prime }\in \mathcal{W}$ and $%
A\cup R^{\prime }\in \mathcal{L}$. Writing $a=s(A),$ $b=s(B),$ $r=s(R),$ $%
r^{\prime }=s(R^{\prime }),$ the profiles $x=a+r,$ $y=b+r^{\prime }$ belong
to $\overline{\mathcal{W}}$, the profiles $u=b+r,v=a+r^{\prime }$ \ belong
to $\overline{\mathcal{L}}$, and $x+y=u+v$.

To prove the other direction, let $G=(N,\mathcal{W})$ be a $C$-complete
game, and assume for contradiction that there exist $x,y\in \overline{%
\mathcal{W}}$ and $u,v\in \overline{\mathcal{L}}$ with $x+y=u+v$. Thus, for
any $1\leq t\leq d,$
\begin{equation}
x_{t}+y_{t}=u_{t}+v_{t}  \label{x_t y_t}
\end{equation}

\bigskip

\textbf{Claim.} For every $1\leq t\leq d,$ there are sets $%
A_{t},B_{t}\subseteq N_{t}$ and $R_{t},R_{t}^{\prime }\subseteq
N_{t}\diagdown $ $A_{t}\cup B_{t}$ for which the following holds:

\begin{equation}
\left\vert A_{t}\cup R_{t}\right\vert =x_{t},\text{ }\left\vert B_{t}\cup
R_{t}^{\prime }\right\vert =y_{t},\left\vert B_{t}\cup R_{t}\right\vert
=u_{t},\text{ }\left\vert A_{t}\cup R_{t}^{\prime }\right\vert =v_{t}.
\label{sets}
\end{equation}

\bigskip

\begin{proof}
We can choose $X_{t},Y_{t}\subseteq N_{t}$ with $\left\vert X_{t}\right\vert
=x_{t}$, $\left\vert Y_{t}\right\vert =y_{t}$ and $|X_{t}\setminus
Y_{t}|=\min \{x_{t},n_{t}-y_{t}\}$. This choice also implies that%
\begin{equation}
|Y_{t}\setminus X_{t}|=\left\vert Y_{t}\right\vert -\left\vert
X_{t}\right\vert +|X_{t}\setminus Y_{t}|\text{ }=y_{t}-x_{t}+\min
\{x_{t},n_{t}-y_{t}\}=\min \{y_{t},n_{t}-x_{t}\}.  \label{y_t -- x_t}
\end{equation}%
In what follows, we separate the cases of $x_{t}\geq u_{t}$ and $%
x_{t}<u_{t}. $

Assume first that $x_{t}\geq u_{t},$ and put $a_{t}=x_{t}-u_{t}.$ Since $%
x_{t}-u_{t}\;=v_{t}-y_{t}\;\leq \;n_{t}-y_{t}$ by (\ref{x_t y_t}), we have
\begin{equation*}
a_{t}\leq \min \{x_{t},n_{t}-y_{t}\}=|X_{t}\setminus Y_{t}|.
\end{equation*}%
Thus, we can choose $A_{t}\subseteq X_{t}\setminus Y_{t}$ with $\left\vert
A_{t}\right\vert =a_{t}$. Furthermore, let $B_{t}=\emptyset ,$ $%
R_{t}=X_{t}\setminus A_{t}$ and $R_{t}^{\prime }=Y_{t}$ (the two latter sets
are clearly disjoint from $A_{t}\cup B_{t}$ $=A_{t}$). Then:%
\begin{equation*}
\left\vert A_{t}\cup R_{t}^{\prime }\right\vert =\left\vert A_{t}\right\vert
+\left\vert Y_{t}\right\vert =\left( x_{t}-u_{t}\right) +y_{t}\overset{\text{%
(\ref{x_t y_t})}}{=}v_{t}.
\end{equation*}%
We also have%
\begin{equation*}
\left\vert B_{t}\cup R_{t}\right\vert =\left\vert R_{t}\right\vert
=\left\vert X_{t}\setminus A_{t}\right\vert =x_{t}-a_{t}=u_{t},
\end{equation*}%
and $\left\vert B_{t}\cup R_{t}^{\prime }\right\vert =\left\vert
Y_{t}\right\vert =y_{t},$ $\left\vert A_{t}\cup R_{t}\right\vert =\left\vert
X_{t}\right\vert =x_{t}.$ Therefore (\ref{sets}) holds.

Now assume that $x_{t}<u_{t},$ and put $b_{t}=u_{t}-x_{t}.$ Then, from $%
u_{t}\leq x_{t}+y_{t}$ we have%
\begin{equation*}
b_{t}\leq \min \{y_{t},n_{t}-x_{t}\}\overset{\text{(\ref{y_t -- x_t})}}{=}%
|Y_{t}\setminus X_{t}|.
\end{equation*}%
Thus, we can choose $B_{t}\subseteq Y_{t}\setminus X_{t}$ with $\left\vert
B_{t}\right\vert =b_{t}$. Furthermore, let $A_{t}=\emptyset ,$ $R_{t}=X_{t}$
and $R_{t}^{\prime }=Y_{t}\setminus B_{t}$ (the two latter sets are clearly
disjoint from $A_{t}\cup B_{t}$ $=B_{t}$). Then:%
\begin{equation*}
\left\vert A_{t}\cup R_{t}^{\prime }\right\vert =\left\vert R_{t}^{\prime
}\right\vert =\left\vert Y_{t}\setminus B_{t}\right\vert
=y_{t}-b_{t}=y_{t}-\left( u_{t}-x_{t}\right) \overset{\text{(\ref{x_t y_t})}}%
{=}v_{t}.
\end{equation*}%
We also have%
\begin{equation*}
\left\vert B_{t}\cup R_{t}\right\vert =\left\vert B_{t}\right\vert
+\left\vert X_{t}\right\vert =\left( u_{t}-x_{t}\right) +x_{t}=u_{t},
\end{equation*}%
and $\left\vert A_{t}\cup R_{t}\right\vert =\left\vert X_{t}\right\vert
=x_{t}$, $\left\vert B_{t}\cup R_{t}^{\prime }\right\vert =\left\vert
Y_{t}\right\vert =y_{t}.$ Thus, as in the previous case, (\ref{sets}) holds
as well.
\end{proof}

Given the sets constructed in the Claim, define%
\begin{equation*}
A=\cup _{1\leq t\leq d}A_{t},B=\cup _{1\leq t\leq d}B_{t},R=\cup _{1\leq
t\leq d}R_{t},R^{\prime }=\cup _{1\leq t\leq d}R_{t}^{\prime }.
\end{equation*}%
By (\ref{sets}) and the choice of $t$-sets as subsets of $N_{t},$ for every $%
t,$ we obtain
\begin{equation*}
s\left( A\cup R\right) =x,\text{ }s\left( A\cup R^{\prime }\right) =v,\text{
}s\left( B\cup R\right) =u,\text{ }s\left( B\cup R^{\prime }\right) =y.
\end{equation*}%
Since $R,R^{\prime }$ are clearly disjoint from $A\cup B,$ and by assumption
$x,y\in \overline{\mathcal{W}}$ and $u,v\in \overline{\mathcal{L}}$, the
coalitions $A$ and $B$ are $\succeq _{C}$-incomparable: $A\cup R\in $ $%
\mathcal{W},$ $B\cup R$ $\in \mathcal{L}$ and yet $B\cup R^{\prime }$ $\in
\mathcal{W}$ and $A\cup R^{\prime }\in \mathcal{L}$. This contradicts the
assumed $C$-completeness of $(N,\mathcal{W}).$

\subsection{Proof of Proposition \protect\ref{lem:minimax}}

\label{sec proof of minimax}

\textbf{(i)} The function $w\overset{f}{\mapsto }\max_{A\in \mathcal{W}%
,\,B\in \mathcal{L}}\left[ w(B)-w(A)\right] $ is \emph{convex} on $\Delta
(N) $ as a maximum of linear functions. Also, for each $w\in \Delta (N)$ and
an automorphism $\sigma \in \Gamma =$ $\Gamma (G),$ the $\sigma $-translate $%
\sigma w\in \Delta (N)$ (given by $(\sigma w)_{i}=w_{\sigma (i)}$) yields
the same value of $f$ as $w:$%
\begin{eqnarray}
f(\sigma w) &=&\max_{A\in \mathcal{W},B\in \mathcal{L}}\left[ (\sigma
w)(B)-(\sigma w)(A)\right] =\max_{A\in \mathcal{W},B\in \mathcal{L}}\left[
w(\sigma \left( B\right) )-w(\sigma \left( A\right) )\right]
\label{f sigma w} \\
&=&\max_{A\in \sigma \left( \mathcal{W}\right) ,B\in \sigma \left( \mathcal{L%
}\right) }\left[ w(B)-w(A)\right] \overset{\text{(*)}}{=}\max_{A\in \mathcal{%
W},B\in \mathcal{L}}\left[ w(B)-w(A)\right] =f(w),  \label{f sigma w 1}
\end{eqnarray}%
where (*) holds because $\sigma \left( \mathcal{W}\right) =\mathcal{W}$ and $%
\sigma \left( \mathcal{L}\right) =\mathcal{L}$. Thus, replacing a minimizer $%
w\in \Delta (N)$ in (\ref{epsilon **}) by $\bar{w}=\frac{1}{|\Gamma |}%
\sum_{\sigma \in \Gamma }\sigma w\in \Delta (N)$ cannot increase the value $%
\varepsilon ^{\ast }(G)$ of $f$ at $w,$ and so $\bar{w}$ is a minimizer too.
Since $\bar{w}$ is clearly $\Gamma $-invariant ($\sigma \bar{w}=\bar{w}$ for
every $\sigma \in \Gamma $) and, by Proposition \ref{lem:orbits}, $\Gamma $
consists of all permutations that preserve the player classes (and in
particular contains the transpositions of any two given players within the
same class), $\bar{w}$ is constant on each player class.

\textbf{(ii)} As argued above, the search for minimizer in (\ref{epsilon **}%
) can be restricted to class-constant vectors. Therefore the choice of a
weight vector is reducible to selecting $p$ from the standard $\left(
d-1\right) $-simplex $\Delta (d)$ and transforming it into a class-constant $%
w^{p}\in \Delta (N)$ with $w_{i}^{p}=p_{t}/n_{t}$ for each $1\leq t\leq d$
and $i\in N_{t}$ (that is, $p_{t}$ is the total weight of the class $N_{t}$,
divided equally between the members of $N_{t}$). Thus,
\begin{eqnarray*}
\varepsilon ^{\ast }(G) &=&\min_{p\in \Delta (d)}\max_{A\in \mathcal{W}%
,\,B\in \mathcal{L}}\left[ w^{p}(B)-w^{p}(A)\right] =\min_{p\in \Delta
(d)}\max_{A\in \mathcal{W},\,B\in \mathcal{L}}\sum_{t=1}^{d}\frac{p_{t}}{%
n_{t}}\left[ \,s(B)_{t}-s(A)_{t}\right] \\
&&\overset{\text{(*)}}{=}\min_{p\in \Delta (d)}\max_{a\in \operatorname{conv}\,%
\overline{\mathcal{W}},\,b\in \operatorname{conv}\overline{\mathcal{L}}%
}\sum_{t=1}^{d}\frac{p_{t}}{n_{t}}\left[ \,b_{t}-a_{t}\right] \\
&&\overset{\text{(**)}}{=}\max_{a\in \operatorname{conv}\,\overline{\mathcal{W}}%
,\,b\in \operatorname{conv}\overline{\mathcal{L}}}\min_{p\in \Delta
(d)}\sum_{t=1}^{d}p_{t}\frac{b_{t}-a_{t}}{n_{t}} \\
&&\overset{\text{(*)}}{=}\max_{a\in \operatorname{conv}\,\overline{\mathcal{W}}%
,\,b\in \operatorname{conv}\overline{\mathcal{L}}}\left( \min_{1\leq t\leq d}\frac{%
b_{t}-a_{t}}{n_{t}}\right) ,
\end{eqnarray*}%
where (*) holds in both instances because a linear function attains its
minimum or maximum over a polytope at a vertex, and (**) holds by the
minimax theorem.

\subsection{Proof components of Theorem \protect\ref{thm:main} that are
missing in Section \protect\ref{sec:proof}}

\label{sec missing proof components}

\subsubsection{Proof of Lemma \protect\ref{lem:mid}}

Notice that equality (\ref{sumset eq}) holds for $x,y\in \overline{\mathcal{W%
}}$ and $u=v=(x+y)/2\in \Lambda (\bar{n}).$ Since the game is $C$-complete,
by Proposition \ref{prop:sumset} it cannot be the case that $(x+y)/2\in
\overline{\mathcal{L}}$.

\subsubsection{Proof of Lemma \protect\ref{lem:rmid}}

For any $t$ at which $x_{t}+y_{t}$ is odd, $x_{t}$ must be different from $%
y_{t}$, and so at least one of $x_{t},y_{t}$ is smaller than $n_{t}$;
increase that one by $1$. The two updated profiles, $\overline{x}\in \Lambda
(\bar{n})$ and $\overline{y}\in \Lambda (\bar{n}),$ are still in $\overline{%
\mathcal{W}}$ because $\overline{\mathcal{W}}$ is up-closed. Also, clearly $%
\left( \overline{x}+\overline{y}\right) /2=\lceil (x+y)/2\rceil $, and
therefore by Lemma \ref{lem:mid} this profile is in $\overline{\mathcal{W}}$%
. The statement for $\overline{\mathcal{L}}$ is obtained symmetrically, by
decreasing one component that is greater than $0$ of any odd $x_{t}+y_{t}$,
and using down-closedness of $\overline{\mathcal{L}}$.

\subsubsection{Proof of Lemma \protect\ref{lem:round}}

\label{sec lemma round}

For the supremum norm $\left\Vert \cdot \right\Vert _{\infty }$ on $\mathbb{R%
}^{d},$ let%
\begin{equation*}
\mathbf{dist}_{\infty }(a,A)=\inf \left\{ \left\Vert a-b\right\Vert _{\infty
}:b\in A\right\}
\end{equation*}%
be the $\left\Vert \cdot \right\Vert _{\infty }$-distance between a vector $%
a\in \mathbb{R}^{d}$ and a set $A\subset \mathbb{R}^{d}.$ By Lemma \ref%
{lem:rmid}, for any $x,y\in \overline{\mathcal{W}}$ we have $\lceil
(x+y)/2\rceil \in \overline{\mathcal{W}}$ , and so%
\begin{equation}
\mathbf{dist}_{\infty }((x+y)/2,\overline{\mathcal{W}})\leq \left\Vert
\lceil (x+y)/2\rceil -(x+y)/2\right\Vert _{\infty }\leq 1/2.
\label{dist<=1/2}
\end{equation}

By Dilworth et al. \hyperlink{cite.DHR}{(1999}, Theorem 1), for any $%
A\subset \mathbb{R}^{d}$ and $\xi >0$ the following holds: if $\mathbf{dist}%
_{\infty }((x+y)/2,A)$ $\leq $ $\xi $ for any $x,y\in A\mathcal{\ }$then for
any $a\in \operatorname{conv}A,$
\begin{equation}
\mathbf{dist}_{\infty }(a,A)\leq \xi \kappa \left( d\right)  \label{kappa1}
\end{equation}%
for a sharp $\kappa \left( d\right) $ that satisfies%
\begin{equation}
\kappa \left( d\right) \leq \lfloor \log _{2}d\rfloor +1+\frac{d}{2^{\lfloor
\log _{2}d\rfloor }}.  \label{kappa}
\end{equation}%
(Theorem 1 of Dilworth et al. \hyperlink{cite.DHR}{(1999)} holds with the
same bound $\kappa (d)$ for \emph{any} norm on $\mathbb{R}^{d};$ we use the
supremum norm $\left\Vert \cdot \right\Vert _{\infty }$ on $\mathbb{R}^{d}$
because it allows easy coordinate comparisons.)

By denoting $x=\frac{d}{2^{\lfloor \log _{2}d\rfloor }}\in \lbrack 1,2)$ and
using the inequality $x\leq 1+\log _{2}x$ (that holds for all $x\in \lbrack
1,2]$)$,$ we obtain the following:
\begin{eqnarray*}
\lfloor \log _{2}d\rfloor +1+\frac{d}{2^{\lfloor \log _{2}d\rfloor }}
&=&\lfloor \log _{2}d\rfloor +1+x\leq \lfloor \log _{2}d\rfloor +2+\log _{2}x
\\
&=&\lfloor \log _{2}d\rfloor +2+\log _{2}\left( \frac{d}{2^{\lfloor \log
_{2}d\rfloor }}\right) =\log _{2}d+2.
\end{eqnarray*}%
Therefore, $\kappa \left( d\right) $ in (\ref{kappa}) has a more elegant
(though less tight) bound,
\begin{equation}
\kappa \left( d\right) \leq \log _{2}d+2.  \label{kappa3}
\end{equation}%
Applying (\ref{kappa1}) to $A=\overline{\mathcal{W}}$ and using $\xi =1/2$
(which is due to (\ref{dist<=1/2})), we obtain the following inequality for
any $a\in \operatorname{conv}\overline{\mathcal{W}}$:

\begin{equation*}
\mathbf{dist}_{\infty }(a,\overline{\mathcal{W}})\overset{\text{(\ref{kappa1}%
)}}{\leq }\kappa \left( d\right) /2\overset{\text{(\ref{kappa3})}}{\leq }%
\left( \log _{2}d+2\right) /2\text{ }\overset{\text{(\ref{delta})}}{=}\gamma
\text{.}
\end{equation*}%
Accordingly, there exists $u\in \overline{\mathcal{W}}$ with $\left\Vert
a-u\right\Vert _{\infty }\leq \gamma ,$ which clearly implies $u\leq
a+\gamma \mathbf{1.}$ The statement for $\overline{\mathcal{L}}$ is obtained
similarly, by using the claim for $\overline{\mathcal{L}}$ in Lemma \ref%
{lem:rmid}.

\subsubsection{Proof of the lower bound in inequality (\protect\ref{epsilon
star G}) of Theorem \protect\ref{thm:main}}

\label{sec lower bound}

Here we establish the lower bound for $\varepsilon ^{\ast }(G)$. By
Proposition \ref{lem:minimax}(i), the minimum in (\ref{epsilon **}) is
attained at some $w^{\ast }\in \Delta (N)$ that is constant on player
classes; for such a $w^{\ast }$, for any $1\leq t\leq d$ and any $i\in N_{t}$%
,
\begin{equation}
w_{i}^{\ast }=w^{\ast }(N_{t})/n_{t}\leq 1/n_{t}\leq 1/n_{\min }.
\label{w*<=1/n_min}
\end{equation}%
Let $A^{\ast }\in \mathcal{W}$ be the coalition attaining $\min_{A\in
\mathcal{W}}w^{\ast }(A)$. If $w_{i}^{\ast }=0$ for every $i\in A^{\ast }$
then $\min_{A\in \mathcal{W}}w^{\ast }(A)=$ $0$ $\leq $ $\max_{B\in \mathcal{%
L}}w^{\ast }(B),$ and so $\varepsilon ^{\ast }(G)\geq 0$. In the
complementary case there exists $i\in A^{\ast }$ with $w_{i}^{\ast }>0$, and
then $w^{\ast }(A^{\ast }\setminus \{i\})<w^{\ast }(A^{\ast })$, implying
that $A^{\ast }\setminus \{i\}$ is losing by the minimality of $w^{\ast
}(A^{\ast })$. But then
\begin{equation*}
\max_{B\in \mathcal{L}}w^{\ast }(B)\ \geq \ w^{\ast }(A^{\ast })-w_{i}^{\ast
}\ \overset{\text{(\ref{w*<=1/n_min})}}{\geq }\ \min_{A\in \mathcal{W}%
}w^{\ast }(A)-1/n_{\min },
\end{equation*}%
and so, by the choice of $w^{\ast },$
\begin{equation*}
\varepsilon ^{\ast }(G)=\max_{B\in \mathcal{L}}w^{\ast }(B)-\min_{A\in
\mathcal{W}}w^{\ast }(A)\geq -1/n_{\min }.
\end{equation*}

\subsection{Proof of Proposition \protect\ref{prop gappoed weighted games}}

To check $C$-completeness, suppose that there exists a $2$-trade, namely,
for some $X_{1},X_{2}\in \mathcal{W}$ and $Y_{1},Y_{2}\in \mathcal{L}$ the
following holds:
\begin{equation*}
\chi _{X_{1}}+\chi _{X_{2}}\;=\;\chi _{Y_{1}}+\chi _{Y_{2}}.
\end{equation*}%
We thus have
\begin{equation}
W(X_{1})+W(X_{2})=W(Y_{1})+W(Y_{2}).  \label{2 trade in gapped weight}
\end{equation}%
Since $W(X_{1}),W(X_{2})\neq q+1$ $\ $(the weight $q+1,$ by (iv) in
Definition \ref{def punctured games}, is not attained by any coalition),
(ii) in Definition \ref{def punctured games} implies that $%
W(X_{1}),W(X_{2})\in $ $\{q-1\}\cup \lbrack q+2,\infty )$, and so%
\begin{equation}
W(X_{1})+W(X_{2})\in \{2q-2\}\cup \lbrack 2q+1,\infty ).  \label{contr1}
\end{equation}%
Similarly, since $W(Y_{1}),W(Y_{2})\neq q-2$ $\ $(the weight $q-2$, by (iv)
in Definition \ref{def punctured games}, is not attained by any coalition),
(ii) in Definition \ref{def punctured games} implies that $%
W(Y_{1}),W(Y_{2})\in $ $[0,q-3]\cup \{q\},$ and so
\begin{equation}
W(Y_{1})+W(Y_{2})\in \lbrack 0,2q-3]\cup \{2q\}.  \label{contr2}
\end{equation}%
But the sets $\{2q-2\}\cup \lbrack 2q+1,\infty )$ and $[0,2q-3]\cup \{2q\}$
are disjoint, and so (\ref{2 trade in gapped weight}) is contradicted by the
combination of (\ref{contr1}) and (\ref{contr2}). Consequently, $G$ is $C$%
-complete by Proposition \ref{lem:inc}.

Next, take $w=W/W(N)\in \Delta (N).$ By (iii) in Definition \ref{def
punctured games}, the weights $q-1$ and $q$ are attained by some coalitions,
and so $\min_{A\in \mathcal{W}}w(A)=(q-1)/W(N)$ and $\max_{B\in \mathcal{L}%
}w(B)=q/W(N)$. The inequality (\ref{epslilon*<=w(n)}) then follows from (\ref%
{epsilon **}) in Definition \ref{de loss win gap}.

Finally, if there exist $a\in \operatorname{conv}\overline{\mathcal{W}},$ $b\in
\operatorname{conv}\,\overline{\mathcal{L}}$ such that $b-a=\bar{n}/W(N)$ then $%
\varepsilon ^{\ast }(G)\geq 1/W(N)$ by (\ref{eq:mm}) in Proposition \ref%
{lem:minimax}, which implies equality in (\ref{epslilon*<=w(n)}).

\subsection{Proof of Proposition \protect\ref{Games with big theta}}

\textbf{\noindent }In $\widetilde{G}_{m},$ coalitions that weigh exactly $%
q-1=59$ are precisely those consisting of one $N_{1}$-player, one $N_{2}$%
-player and two $N_{3}$-players, and coalitions that weigh $q=60$ are
precisely those consisting of either three $N_{1}$-players, or four $N_{2}$%
-players, or five $N_{3}$-players (such coalitions exist because, by
assumption, $n_{\min }=$ $m$ $\geq 5.$). Also, \emph{there are no coalitions
weighing} $q+1=61$ or $q-2=58,$ as is easy to verify. Thus, the game $%
\widetilde{G}_{m}$ satisfies (iii) and (iv) in Definition \ref{def punctured
games} (in addition to (i) and (ii) that follow from the premise of the
proposition), and so $\widetilde{G}_{m}$ is a flip game for the given $W$
and $q=60$. In particular, the game is $C$-complete and $\varepsilon ^{\ast
}(\widetilde{G}_{m})\leq 1/W(N)=1/\left( 47m\right) $ by Proposition \ref%
{prop gappoed weighted games}.

$C$-completeness of $\widetilde{G}_{m}$ implies its $I$-completeness, and
its three $I$-desirability classes are precisely the symmetry classes $%
N_{1},N_{2},N_{3}$ by Proposition \ref{lem:orbits}. In particular, the class
sizes are $\overline{n}=(m,m,m)$ (with $m\geq 5$ by assumption) and $n_{\min
}=m$.

We now show that $\varepsilon ^{\ast }(\widetilde{G}_{m})$ $=1/\left(
47m\right) $. \textit{\ }Take $a=(1,1,2)\in \overline{\mathcal{W}}$ (this
profile corresponds to a winning coalition with one $N_{1}$-player, one $%
N_{2}$-player, and two $N_{3}$-players), and
\begin{equation*}
b=\tfrac{1}{47}(48,48,95)=\tfrac{16}{47}(3,0,0)+\tfrac{12}{47}(0,4,0)+\tfrac{%
19}{47}(0,0,5)\in \operatorname{conv}\overline{\mathcal{L}}.
\end{equation*}%
The profile $b$ is a convex combination of three coalition profiles: $%
(3,0,0)\in \,\overline{\mathcal{L}}$ (corresponding to a losing coalition
with three $N_{1}$-players), $(0,4,0)\in \,\overline{\mathcal{L}}$
(corresponding to a losing coalition with four $N_{2}$-players) and $%
(0,0,5)\in \,\overline{\mathcal{L}}$ (corresponding to a losing coalition
with five $N_{3}$-players). Then%
\begin{equation*}
b-a=\frac{1}{47}(1,1,1)=\frac{1}{47m}(m,m,m)=\frac{1}{47m}\bar{n}=\bar{n}%
/W(N),
\end{equation*}%
and the second part of Proposition \ref{prop gappoed weighted games} gives $%
\varepsilon ^{\ast }\left( \widetilde{G}_{m}\right) =1/W(N)$ $=1/(47m)$.

\subsection{Proof of Proposition \protect\ref{acyclicity and
non-weightedness}}

Fix $m\geq 5,$ and consider the game $\widetilde{G}_{m}.$ Call a coalition
\emph{sub-}$\emph{59}$ if it is contained in a coalition of weight $59$.
Thus, a sub-$59$ coalition profile is weakly dominated by $(1,1,2),$ and its
weight lies in the set
\begin{equation}
S=\{0,12,15,20,24,27,32,35,39,44,47,59\}.  \label{no consequtive integers}
\end{equation}%
We begin with the following partial characterization of the $C$-desirability
relation $A\succeq _{C}B$ and of its strict part by the relation between $%
W(A)$ and $W(B).$

\begin{lemma}
\label{lemma ascent descent1}Let $A,B$ be coalitions in the game $\widetilde{%
G}_{m}$.

\begin{itemize}
\item[(a)] If $W(B)=W(A)$ then $A\sim _{C}B$.

\item[(b)] If $W(B)\geq W(A)+2$ then $B\succeq _{C}A$.

\item[(c)] If $W(B)=W(A)+1$ then $A\succeq _{C}B$; and $A\succ _{C}B$ iff $%
W(A\cup R)=59$ for some $R\subseteq N\setminus (A\cup B)$.
\end{itemize}
\end{lemma}

\begin{proof}
In what follows, $R$ will denote a generic subset of $N\setminus (A\cup B)$.
Clearly, $W(A\cup R)=W(A)+W(R)$ and $W(B\cup R)=W(B)+W(R)$. We next compare
the winning status of $A\cup R$ and $B\cup R$ in the three cases considered
in the lemma.

(a) Since $W(A\cup R)=W(B\cup R),$ the coalitions $A\cup R$ and $B\cup R$
have the same winning status. Consequently, $A\sim _{C}B.$

(b) If $A\cup R\in \mathcal{W}$ then either $W(A\cup R)=59$ or $W(A\cup
R)\geq 62$, and so $W(B\cup R)\geq 61$, hence $W(B\cup R)\geq 62$ (because $%
61$ is never attained). Therefore, $B\cup R\in \mathcal{W}$, and so $%
B\succeq _{C}A$.

(c) If $B\cup R\in \mathcal{W}$ with $W(B\cup R)=59$, then $W(A\cup R)=58$,
which is an unattainable value in our flip game. If $B\cup R\in \mathcal{W}$
with $W(B\cup R)\geq 62$, then $W(A\cup R)\geq 61$, hence $W(A\cup R)\geq 62$%
, and so $A\cup R\in \mathcal{W}$. It follows that $A\succeq _{C}B$.
Additionally, $A\succ _{C}B$ iff, for some $R,$ $A\cup R\in \mathcal{W}$ but
$B\cup R\in \mathcal{L}$. It cannot be that $W(A\cup R)\geq 62$ because then
$W(B\cup R)=W(A\cup R)+1\geq 63$ and $B\cup R$ would be winning too, so it
must be the case that $W(A\cup R)=59$.
\end{proof}

\begin{corollary}
\label{lemma ascent descent}If $A\succ _{C}B$ then either $W(B)\leq W(A)-2$
(call this a \emph{drop}) or $W(B)=W(A)+1$ and $A$ is a sub-$59$ coalition
(call this a \emph{rise}).
\end{corollary}

\begin{proof}
By Lemma \ref{lemma ascent descent1}(a), $A\succ _{C}B$ implies that $%
W(A)\neq W(B)$ . If $W(A)=W(B)+1$ then $B\succeq _{C}A$ by Lemma \ref{lemma
ascent descent1}(c), and so $A\succ _{C}B$ is impossible. If $W(A)\leq
W(B)-2 $ then $B\succeq _{C}A$ by Lemma \ref{lemma ascent descent1}(b),
which again makes $A\succ _{C}B$ impossible. The remaining cases are $%
W(A)\geq W(B)+2$ and $W(A)=W(B)-1$; in the latter case, the strict dominance
$A\succ _{C}B$ implies, by Lemma \ref{lemma ascent descent1}(c), that for
some $R\subseteq N\setminus (A\cup B)$ we have $W(A\cup R)=59$, i.e.,\ $A$
is sub-$59$ coalition.
\end{proof}

We now prove Proposition \ref{acyclicity and non-weightedness}. Suppose that
there exists a $\succ _{C}$-cycle $X_{1}\succ _{C}X_{2}\succ _{C}\dots \succ
_{C}X_{r}\succ _{C}X_{1}$. By Corollary \ref{lemma ascent descent}, when we
move along the cycle, the weight of a coalition either \emph{rises} or \emph{%
drops} (in the sense defined in that corollary) at each transition, with the
total weight change being $0$. With $k$ rises (of one weight unit) and $l$
drops (of at least two weight units) this forces $k\geq 2l$ and $k\geq 1$
(having fewer than twice as many rises as drops would make the total weight
change negative, and a cycle of drops only would do the same). Now, two
rises cannot be consecutive: a rise $X\succ _{C}Y$ followed by another rise $%
Y\succ _{C}Z$ would make both $X$ and $Y$ \emph{sub-}$\emph{59}$, with $%
W(Y)=W(X)+1$, but note that $S$ in (\ref{no consequtive integers}) contains
\emph{no two consecutive integers}. Hence, along the cycle every rise is
followed by a drop, giving $l\geq k$. Together, we have $k\geq 2l\geq 2k$,
and so $k=0$, a contradiction.

\subsection{Proof of the claims made in Remark \protect\ref{rem constant-sum
acyclic} on the $14$-player game (a sketch)}

\label{sec sketch}

The following pair,
\begin{equation*}
a=\left[ 288\,(1,2,3,1)+19\,(1,3,0,4)\right] /307\in \operatorname{conv}\overline{%
\mathcal{W}}
\end{equation*}%
and%
\begin{equation*}
b=\left[ 92\,(0,3,4,4)+120\,(1,3,1,0)+95\,(2,0,4,0)\right] /307\in \operatorname{conv%
}\overline{\mathcal{L}},
\end{equation*}%
satisfies $b-a=(3,3,4,4)/307=\overline{n}/W(N)$, and so $\varepsilon ^{\ast
}=1/307$ by the last part of Proposition \ref{prop gappoed weighted games}.
Acyclicity is proved as in Proposition \ref{acyclicity and non-weightedness}%
. Similarly to Corollary \ref{lemma ascent descent}, it can be shown that $%
A\succ _{C}B$ implies either a \emph{drop} $W(B)\leq W(A)-2$ (with equality
only if $A$ is contained in a coalition of weight $153$ or $156,$ call such $%
A$ a sub-$153$ or a sub-$156$ coalition), or a \emph{rise} $W(B)=W(A)+1$
only if $A$ is contained in a coalition of weight $153$ (i.e., $A$ is sub-$%
153$). The weights of sub-$153$ coalitions may include two consecutive
integers (e.g., $12$ and $13$); however, if some consecutive integers $u$
and $u+1$ are weights of sub-$153$ coalitions then $u+2$ cannot be the
weight of either a sub-$153$ or a sub-$156$ coalition (the weight-$156$
coalition profiles are $(1,2,3,2)$ and $(3,0,0,1)$). Hence, along a presumed
$\succ _{C}$-cycle, every sequence of consecutive rises has length at most $2
$ and is followed by a drop. Because the total change in weight along a
cycle is zero, rises must come in sequences of length $2$ and every drop
equals $2$, which the above property of $u$ and $u+1$ excludes.

\subsection{Proof of Proposition \protect\ref{relation between thresholds}}

Since any rescaling of the weight vector $W$ does not change the ratio $%
W(B)/W(A),$ and, clearly, $\min_{A\in \mathcal{W}}W(A)>0$ for any optimal
solution $W$ in (\ref{mu(G) original}), it may be assumed w.l.o.g. that the
set of minimizers in (\ref{mu(G) original}) is restricted to be%
\begin{equation}
\widehat{\Delta }(G)=\left\{ W\in \mathbb{R}_{+}^{n}:\min_{A\in \mathcal{W}%
}W(A)=1\right\} .  \label{delta hat}
\end{equation}%
That is, the following alternative expression can be used for $\mu (G):$%
\begin{equation}
\mu (G)=\max \left\{ 1,\min_{W\in \widehat{\Delta }(G)}\max_{B\in \mathcal{L}%
}W(B)\right\} .  \label{mu(G)}
\end{equation}

We begin by establishing (\ref{comparing epsilon * and mu1}). Fix some $%
W^{\ast }\in \widehat{\Delta }(G)$ at which $\alpha ^{\ast }=\min_{W\in
\widehat{\Delta }(G)}\max_{B\in \mathcal{L}}W(B)$ is attained (here, it is
possible that $\alpha ^{\ast }<1$)$.$ Clearly,
\begin{equation}
W^{\ast }(N)\geq \min_{A\in \mathcal{W}}W^{\ast }(A)\overset{\text{(\ref%
{delta hat})}}{=}1,  \label{W*(N)>=1}
\end{equation}%
and so for a normalized $w^{\ast }=W^{\ast }/W^{\ast }(N)$ $\in \Delta
\left( N\right) $ we have
\begin{equation*}
\alpha ^{\ast }/W^{\ast }(N)=\max_{B\in \mathcal{L}}w^{\ast }(B)\geq
\min_{A\in \mathcal{W}}w^{\ast }(A)+\varepsilon ^{\ast }\left( G\right)
=1/W^{\ast }(N)+\varepsilon ^{\ast }\left( G\right) .
\end{equation*}%
Therefore, when $\varepsilon ^{\ast }\left( G\right) \geq 0,$%
\begin{equation*}
\mu (G)\overset{\text{(\ref{mu(G)})}}{=}\max \{\alpha ^{\ast },1\}\geq
1+W^{\ast }(N)\varepsilon ^{\ast }\left( G\right) \overset{\text{(\ref%
{W*(N)>=1})}}{\geq }1+\varepsilon ^{\ast }\left( G\right) ,
\end{equation*}%
proving (\ref{comparing epsilon * and mu1})$.$ When $\varepsilon ^{\ast
}\left( G\right) <0,$ (\ref{comparing epsilon * and mu1}) holds trivially
because $\mu (G)\geq 1$ by assumption.

As for (\ref{comparing epsilon * and mu}), consider a proper $G$. Fix some $%
w^{\ast }\in \Delta (N)$ attaining the minimum in the expression (\ref%
{epsilon **}) for $\varepsilon ^{\ast }(G)$, and let $A^{\ast }\in \mathcal{W%
}$ be the coalition attaining the minimum in%
\begin{equation}
q^{\ast }:=\min_{A\in \mathcal{W}}w^{\ast }(A).  \label{q* use later}
\end{equation}%
Then
\begin{equation}
\max_{B\in \mathcal{L}}w^{\ast }(B)=q^{\ast }+\varepsilon ^{\ast }(G).
\label{q* 1}
\end{equation}%
By the properness of $G,$ $N\setminus A^{\ast }\in \mathcal{L}$, and so%
\begin{equation}
1-q^{\ast }=w^{\ast }(N\setminus A^{\ast })\leq q^{\ast }+\varepsilon ^{\ast
}(G)  \label{q*}
\end{equation}%
and%
\begin{equation}
q^{\ast }\geq (1-\varepsilon ^{\ast }(G))/2>0  \label{q*_2}
\end{equation}%
(the last strict inequality is due to the assumption that $\varepsilon
^{\ast }(G)\in \lbrack 0,1)$). The rescaled weight vector $W^{\ast }=\frac{%
w^{\ast }}{q^{\ast }}$ belongs to $\widehat{\Delta }(G)$ and satisfies
\begin{eqnarray}
\min_{W\in \widehat{\Delta }(G)}\max_{B\in \mathcal{L}}W(B) &\leq
&\max_{B\in \mathcal{L}}W^{\ast }(B)\overset{\text{(\ref{q* 1})}}{=}%
1+\varepsilon ^{\ast }(G)/q^{\ast }  \label{use again in last prop} \\
\overset{\text{(\ref{q*_2})}}{\leq }1+\varepsilon ^{\ast }(G)/\left[
(1-\varepsilon ^{\ast }(G))/2\right]  &=&\left( 1+\varepsilon ^{\ast
}(G)\right) /\left( 1-\varepsilon ^{\ast }(G)\right) .
\label{use again in last prop2}
\end{eqnarray}%
Thus, (\ref{comparing epsilon * and mu}) is implied:
\begin{equation*}
\mu \left( G\right) \overset{\text{(\ref{mu(G)})}}{=}\max \left\{
1,\min_{W\in \widehat{\Delta }(G)}\max_{B\in \mathcal{L}}W(B)\right\}
\overset{\text{(\ref{use again in last prop})-(\ref{use again in last prop2})%
}}{\leq }\left( 1+\varepsilon ^{\ast }(G)\right) /\left( 1-\varepsilon
^{\ast }(G)\right) .
\end{equation*}

We now show that the upper bound on $\mu (G)$ in (\ref{comparing epsilon *
and mu}) is also a lower bound on $\mu (G)$ when $G$ is constant-sum.
Assuming that $G$ is constant-sum and using the definition of $\widehat{%
\Delta }(G)$ in (\ref{delta hat}), we have
\begin{equation}
\min_{W\in \widehat{\Delta }(G)}\max_{B\in \mathcal{L}}W(B)=\min_{W\in
\widehat{\Delta }(G)}\left[ W(N)-\min_{A\in \mathcal{W}}W(A)\right]
=\min_{W\in \widehat{\Delta }(G)}W(N)-1.  \label{min W(N)-1}
\end{equation}%
Fix some $W^{\ast }\in \widehat{\Delta }(G)$ that attains $\min_{W\in
\widehat{\Delta }(G)}W(N)$ $>0$. Then, for a normalized $w^{\ast }=W^{\ast
}/W^{\ast }(N)$ $\in \Delta (N),$

\begin{equation}
\min_{A\in \mathcal{W}}w^{\ast }(A)\leq \max_{w\in \Delta (N)}\min_{A\in
\mathcal{W}}w(A)\overset{\text{(\ref{epsilon constant sum})}}{=}%
1/2-\varepsilon ^{\ast }\left( G\right) /2,  \label{use of cs}
\end{equation}%
and $\ $%
\begin{equation}
1=\min_{A\in \mathcal{W}}W^{\ast }(A)=W^{\ast }(N)\min_{A\in \mathcal{W}%
}w^{\ast }(A).  \label{W*(N)}
\end{equation}%
Consequently,

\begin{equation}
W^{\ast }(N)\overset{\text{(\ref{W*(N)})}}{=}\frac{1}{\min_{A\in \mathcal{W}%
}w^{\ast }(A)}\overset{\text{(\ref{use of cs})}}{\geq }2/\left(
1-\varepsilon ^{\ast }\left( G\right) \right) .  \label{W* epsilon*}
\end{equation}%
Then:%
\begin{eqnarray}
\min_{W\in \widehat{\Delta }(G)}\max_{B\in \mathcal{L}}W(B)\overset{\text{(%
\ref{min W(N)-1})}}{=}\min_{W\in \widehat{\Delta }(G)}W(N)-1 &\geq &W^{\ast
}(N)-1  \label{mu(G) cs} \\
\overset{\text{(\ref{W* epsilon*})}}{\geq }2/\left( 1-\varepsilon ^{\ast
}\left( G\right) \right) -1 &=&\left( 1+\varepsilon ^{\ast }(G)\right)
/\left( 1-\varepsilon ^{\ast }(G)\right) .  \label{mu(G) cs1}
\end{eqnarray}%
We finally have
\begin{equation*}
\mu \left( G\right) \overset{\text{(\ref{mu(G)})}}{=}\max \left\{
1,\min_{W\in \widehat{\Delta }(G)}\max_{B\in \mathcal{L}}W(B)\right\}
\overset{\text{(\ref{mu(G) cs})-(\ref{mu(G) cs1})}}{\geq }\left(
1+\varepsilon ^{\ast }(G)\right) /\left( 1-\varepsilon ^{\ast }(G)\right) .
\end{equation*}%
This shows that (\ref{comparing epsilon * and mu}) holds with equality.

\subsection{Proof of Proposition \protect\ref{divergence of mu and epsilon*}}

Fix $m\geq 5,$ and denote $\widetilde{G}_{m}$ throughout by $G=(N,\mathcal{W}%
).$ Suppose that $G$ is $\alpha $-roughly weighted for some $\alpha \geq 1,$
and let $W^{\prime }\in \mathbb{R}_{+}^{N}$ be a weight vector satisfying (%
\ref{alpha inequalities}) in Definition \ref{alpha definition}, i.e.,%
\begin{equation}
A\in \mathcal{W}\ \Longrightarrow \ W^{\prime }(A)\ \geq 1
\label{more than one}
\end{equation}%
and
\begin{equation}
B\in \mathcal{L}\ \Longrightarrow \ W^{\prime }(B)\ \leq \alpha .
\label{less than alpha}
\end{equation}%
Because $\sigma \left( \mathcal{W}\right) =\mathcal{W}$ and $\sigma \left(
\mathcal{L}\right) =\mathcal{L}$ for any automorphism $\sigma \in \Gamma
=\Gamma (G)$ of the game, the above $W^{\prime }$ can be replaced by its $%
\sigma $-translate $\sigma W^{\prime }$ defined by $\left( \sigma W^{\prime
}\right) _{i}=W_{\sigma (i)}^{\prime }$, for any $\sigma \in \Gamma ,$ with (%
\ref{more than one}) and (\ref{less than alpha}) preserved. Given that, $%
W^{\prime }$ can in fact be replaced by the average of its $\sigma $%
-translates over all $\sigma \in \Gamma ,$ with (\ref{more than one}) and (%
\ref{less than alpha}) still preserved. The average $\overline{W}$ is $%
\Gamma $-invariant (i.e., $\sigma \overline{W}=\overline{W}$), and since, by
Proposition \ref{lem:orbits}, $\Gamma $ consists of all possible
permutations that preserve the player classes$,$ $\overline{W}$ is in fact
\emph{constant} on each $N_{t}$, with $\overline{W}_{i}=\frac{\overline{W}%
(N_{t})}{m}$ for any $1\leq t\leq 3$ and $i\in N_{t}.$

Now consider some $A_{1}\in \mathcal{W}$ consisting of exactly one $N_{1}$%
-player, one $N_{2}$-player, and two $N_{3}$-players, and three losing
coalitions: $B_{1}\in \mathcal{L}$ that consists of exactly three $N_{1}$%
-players, $B_{2}\in \mathcal{L}$ that consists of exactly four $N_{2}$%
-players, and $B_{3}\in \mathcal{L}$ consisting of exactly five $N_{3}$%
-players. Thus,%
\begin{equation*}
\overline{W}(B_{1})=3\overline{W}(N_{1})/m\overset{\text{(\ref{less than
alpha})}}{\leq }\alpha ,\text{ \ }\overline{W}(B_{2})=4\overline{W}(N_{2})/m%
\overset{\text{(\ref{less than alpha})}}{\leq }\alpha ,\text{ \ }\overline{W}%
(B_{3})=5\overline{W}(N_{3})/m\overset{\text{(\ref{less than alpha})}}{\leq }%
\alpha ,
\end{equation*}%
and so
\begin{equation*}
1\overset{\text{(\ref{more than one})}}{\leq }\overline{W}(A_{1})=\left(
\overline{W}(N_{1})+\overline{W}(N_{2})+2\overline{W}(N_{3})\right) /m\leq
\left( \frac{1}{3}+\frac{1}{4}+\frac{2}{5}\right) \alpha =\frac{59}{60}%
\alpha .
\end{equation*}%
Consequently, $\alpha \geq 60/59.$ This holds for any $\alpha $ for which $G$
is $\alpha $-roughly weighted, and hence $\mu \left( G\right) \geq 60/59$.
On the other hand, for the weight vector $W$ in Definition (\ref{Family G
tilda}), $W^{\prime }=W/59$ satisfies (\ref{more than one}) and also $%
W^{\prime }(B)\leq 60/59$ for any $B\in \mathcal{L}$, and so $G$ is $60/59$%
-roughly weighted. We conclude that $\mu (\widetilde{G}_{m})=\mu \left(
G\right) =60/59$.

\subsection{Proof of Proposition \protect\ref{Corollary non proper}}

Let $q_{0}=\frac{1}{2}\lim \inf_{n\rightarrow \infty }q_{n}^{\ast }>0,$ and
let $n_{0}$ be such that $q_{n}^{\ast }\geq q_{0}$ for each $n\geq n_{0}.$
Fix $n\geq n_{0},$ and denote $G_{n}$ for as long as needed by $G=(N,%
\mathcal{W}).$ We adopt some notation and facts established in the proof of
Proposition \ref{relation between thresholds} that do not rely on the
assumption made therein that $G$ is proper: these are, specifically,
equations (\ref{delta hat}), (\ref{mu(G)}) and (\ref{q* use later}) for $%
\widehat{\Delta }(G)$, $\mu \left( G\right) $ and $q^{\ast },$ and the
inequality (\ref{use again in last prop}) that applies to all games:

\begin{equation*}
\min_{W\in \widehat{\Delta }(G)}\max_{B\in \mathcal{L}}W(B)\leq
1+\varepsilon ^{\ast }(G)/q^{\ast }.
\end{equation*}%
Thus,%
\begin{equation}
\mu (G)\overset{\text{(\ref{mu(G)})}}{=}\max \left\{ 1,\min_{W\in \widehat{%
\Delta }(G)}\max_{B\in \mathcal{L}}W(B)\right\} \leq 1+\max \{0,\varepsilon
^{\ast }(G)\}/q^{\ast }.  \label{intermediate}
\end{equation}

Since our game is $G=G_{n},$ for which $q_{n}^{\ast }$ is defined in the
same way as $q^{\ast }$ in (\ref{q* use later}) (and hence $q_{n}^{\ast
}=q^{\ast }$), we in fact have the following inequality for every $n\geq
n_{0}$:%
\begin{equation*}
1\leq \mu (G_{n})\overset{\text{(\ref{intermediate})}}{\leq }1+\max
\{0,\varepsilon ^{\ast }(G_{n})\}/q_{n}^{\ast }\leq 1+2\max \{0,\varepsilon
^{\ast }(G_{n})\}/q_{0}.
\end{equation*}%
Thus $\lim_{n\rightarrow \infty }\mu (G_{n})=1$ whenever $\lim_{n\rightarrow
\infty }\varepsilon ^{\ast }(G_{n})=0,$ and $\mu (G_{n})=1+O\!\left(
1/n_{\min }\right) $ whenever $\varepsilon ^{\ast }\left( G_{n}\right)
=\;O\!\left( 1/n_{\min }\right) .$ Given Theorem \ref{thm:types}, this
proves the proposition.

\subsection{Proof of Proposition \protect\ref{prop:misclass}}

\label{proof of prop on misclassification}

Denote by $\mathbf{A}$ a random coalition chosen w.r.t. the uniform
distribution on $2^{N},$ or, equivalently, a coalition in which each player $%
i\in N$ participates (independently of the rest) with probability $1/2$. As
noted immediately after Definition \ref{def misclassification}, $\delta (G)$
can be viewed as the least probability with which a weight system
misclassifies $\mathbf{A}$.

\begin{lemma}
\label{lem:be}Consider a non-zero $W\in \mathbb{R}_{+}^{N}$, and let $%
\left\Vert W\right\Vert _{2}=(\sum_{i\in N}W_{i}^{2})^{1/2}$ and $\rho
=\max_{i\in N}W_{i}/\left\Vert W\right\Vert _{2}$. For every closed interval
$I$ of length $\lambda \geq 0$,
\begin{equation}
\mathbf{\Pr }(W(\mathbf{A})\in I)\leq \sqrt{2/\pi }\,\lambda /\left\Vert
W\right\Vert _{2}+2C_{0}\rho ,  \label{lemma be}
\end{equation}%
where $C_{0}\leq 0.56$ is the absolute constant of the Berry--Esseen theorem
(see, e.g., O'Donnell \hyperlink{cite.odonnell}{(2014}, Section 5.2)).
\end{lemma}

\begin{proof}
One may write $W(\mathbf{A})=\sum_{i\in N}X_{i}$ with $X_{i}=W_{i}\xi _{i}$,
where the $\xi _{i}$ are independent random variables that take the values $0
$ and $1$ with probability $1/2$ each. Then
\begin{equation}
\text{(a) }\mathbf{E}W(\mathbf{A})=W(N)/2;\text{ \ (b) }|X_{i}-\mathbf{E}%
X_{i}|=W_{i}/2;\text{ \ (c) }\mathbf{Var}X_{i}=W_{i}^{2}/4,  \label{var}
\end{equation}%
and
\begin{equation}
\sigma ^{2}=\mathbf{Var}W(\mathbf{A})=\sum_{i\in N}\mathbf{Var}X_{i}\overset{%
\text{(\ref{var}c)}}{=}\left\Vert W\right\Vert _{2}^{2}/4.
\label{sigma squared}
\end{equation}%
Let $Y_{i}=(X_{i}-\mathbf{E}X_{i})/\sigma $ for each $i\in N$; the random
variables $Y_{i}$ are independent and
\begin{equation}
\mathbf{E}Y_{i}=0,\text{ }\sum_{i\in N}\mathbf{Var}Y_{i}\overset{\text{(\ref%
{var}b)}}{=}\sum_{i\in N}\frac{W_{i}^{2}/4}{\sigma ^{2}}\overset{\text{(\ref%
{sigma squared})}}{=}1\text{ .}  \label{EVarY}
\end{equation}%
Also, $\sum_{i\in N}Y_{i}=\frac{W(\mathbf{A})-W(N)/2}{\sigma },$ and hence%
\begin{equation}
\mathbf{\Pr }(W(\mathbf{A})\leq x)=\mathbf{\Pr }\left( \sum_{i\in
N}Y_{i}\leq \left[ x-W(N)/2\right] /\sigma \right) .  \label{P=P}
\end{equation}%
Now let $\Phi $ be the standard normal distribution function. Obviously,%
\begin{eqnarray}
&&\left\vert \mathbf{\Pr }(W(\mathbf{A})\leq x)-\Phi \left( \frac{x-W(N)/2}{%
\sigma }\right) \right\vert   \label{berry esseen} \\
&&\overset{\text{(\ref{P=P})}}{=}\left\vert \mathbf{\Pr }\left( \sum_{i\in
N}Y_{i}\leq \frac{x-W(N)/2}{\sigma }\right) -\Phi \left( \frac{x-W(N)/2}{%
\sigma }\right) \right\vert .  \label{verry esseen 2}
\end{eqnarray}

By the Berry--Esseen (central limit) theorem for independent, not
necessarily identically distributed summands (O'Donnell \hyperlink{cite.odonnell}%
{(2014}, Section 5.2)), applied to the variables $\left\{ Y_{i}\right\}
_{i\in N}$  with $\mathbf{E}Y_{i}$ $=0$ and $\sum_{i\in N}\mathbf{Var}Y_{i}$
$=1$ (by (\ref{EVarY}))$,$ the following holds for all $x\in \mathbb{R}$:%
\begin{equation}
\left\vert \mathbf{\Pr }\left( \sum_{i\in N}Y_{i}\leq \frac{x-W(N)/2}{\sigma
}\right) -\Phi \left( \frac{x-W(N)/2}{\sigma }\right) \right\vert \leq
C_{0}\sum_{i\in N}\mathbf{E}|Y_{i}|^{3}.  \label{verry awwn 3}
\end{equation}%
But also
\begin{equation}
C_{0}\sum_{i\in N}\mathbf{E}|Y_{i}|^{3}=C_{0}\frac{\sum_{i\in N}\mathbf{E}%
\left( |X_{i}-\mathbf{E}X_{i}|^{3}\right) }{\sigma ^{3}}\overset{\text{(\ref%
{var}b),(\ref{sigma squared})}}{=}C_{0}\frac{\sum_{i\in N}W_{i}^{3}}{%
\left\Vert W\right\Vert _{2}^{3}}\overset{\text{(*)}}{\leq }C_{0}\rho ,
\label{verry seen 4}
\end{equation}%
where the last inequality (*) uses $\sum_{i\in N}W_{i}^{3}\leq \left(
\max_{i\in N}W_{i}\right) \left\Vert W\right\Vert _{2}^{2}$. That is, from (%
\ref{berry esseen})--(\ref{verry seen 4}),
\begin{equation}
\left\vert \mathbf{\Pr }(W(\mathbf{A})\leq x)-\Phi \left( \frac{x-W(N)/2}{%
\sigma }\right) \right\vert \leq C_{0}\rho   \label{bery esseen 5}
\end{equation}%
for every $x\in \mathbb{R}$.

Also, for every $x,$ $\mathbf{P}(W(\mathbf{A})<x)=\lim_{y\uparrow x}\mathbf{P%
}(W(\mathbf{A})\leq y)$, and hence (\ref{bery esseen 5}) holds for $\mathbf{P%
}(W(\mathbf{A})<x)$ by the continuity of $\Phi .$

For an interval $I=[a,b]$ with $b-a=\lambda $,
\begin{eqnarray*}
\mathbf{\Pr }(W(\mathbf{A}) &\in &I)=\mathbf{\Pr }(W(\mathbf{A})\leq b)-%
\mathbf{\Pr }(W(\mathbf{A})<a) \\
&&\overset{\text{(\ref{bery esseen 5})}}{\leq }\Phi \left( \frac{b-W(N)/2}{%
\sigma }\right) -\Phi \left( \frac{a-W(N)/2}{\sigma }\right) +2C_{0}\rho  \\
&&\overset{\text{(*)}}{\leq }\frac{\lambda }{\sigma \sqrt{2\pi }}+2C_{0}\rho
\overset{\text{(\ref{sigma squared})}}{=}\sqrt{2/\pi }\,\lambda /\left\Vert
W\right\Vert _{2}+2C_{0}\rho ,
\end{eqnarray*}%
where (*) holds because the two arguments of $\Phi $ differ by $\lambda
/\sigma $ and the standard normal density is at most $1/\sqrt{2\pi }$.
\end{proof}

We now prove Proposition \ref{prop:misclass}. We first make the step of
modifying the definition of $\varepsilon ^{\ast }(G)$ by changing the way a
weight vector $w$ is normalized: consider only weight vectors $W\in \mathbb{R%
}_{+}^{N}$ that are constant on each player class, and whose normalization
is w.r.t. the sum of representative weights from each player class.
Formally, let

\begin{equation*}
\Delta ^{d}(N)=\left\{ W\in \mathbb{R}_{+}^{N}:\exists w\in \Delta (d)\text{
}(W_{i}=w_{t}\text{ whenever }i\in N_{t})\right\} ,
\end{equation*}%
where (recall) $\Delta (d)$ is the standard $\left( d-1\right) $-simplex.
Consider the following variant of the loss-win-excess threshold:
\begin{equation}
\tau (G)=\min_{W\in \Delta ^{d}(N)}\max_{A\in \mathcal{W},B\in \mathcal{L}}%
\left[ W(B)-W(A)\right] .  \label{tau g}
\end{equation}%
Arguing as in the proof of Proposition \ref{lem:minimax}(ii) in Section \ref%
{sec proof of minimax}, we obtain the following:
\begin{equation*}
\tau (G)=\max \left\{ \min_{1\leq t\leq d}\left( b_{t}-a_{t}\right) :\ a\in
\operatorname{conv}\overline{\mathcal{W}},\ b\in \operatorname{conv}\overline{\mathcal{L}}%
\right\} .
\end{equation*}%
Based on the above equality, Lemma \ref{lem:round} and the appropriately
modified argument in the conclusion of the proof of Theorem \ref{thm:main}
in Section \ref{conclusion of the proof of theorem 1} (with $2\gamma $
instead of $2\gamma /n_{\min }$ in (\ref{contradiction delta})) yield:
\begin{equation}
\tau (G)\leq \log _{2}d+2.  \label{tau g2}
\end{equation}

Next, fix some $W^{\ast }\in \Delta ^{d}(N)$ that attains the minimum in (%
\ref{tau g}) and let $q^{\ast }=$ $\min_{A\in \mathcal{W}}W^{\ast }(A)$.
This weight system classifies every winning coalition correctly, and a
losing coalition $A$ is misclassified exactly when $W^{\ast }(A)\geq q^{\ast
}$. Because then $W^{\ast }(A)\leq $ $\max_{B\in \mathcal{L}}W^{\ast }(B)=$ $%
q^{\ast }+\tau (G)$, every misclassified $A$ must satisfy $W^{\ast }(A)\in $
$[q^{\ast },q^{\ast }+\tau ^{+}]$, where $\tau ^{+}=$ $\max \{\tau (G),0\}$.
Hence,%
\begin{equation}
\delta (G)\leq \mathbf{\Pr }(W^{\ast }(\mathbf{A})\in \lbrack q^{\ast
},q^{\ast }+\tau ^{+}]).  \label{deltag prob}
\end{equation}

Since each player in a class $N_{t}$ has the weight $W^{\ast }(N_{t})/n_{t},$
the normalization implied by the definition of $\Delta ^{d}(N)$ means that
\begin{equation}
\sum_{t=1}^{d}W^{\ast }(N_{t})/n_{t}=1.  \label{normalization}
\end{equation}
Thus,%
\begin{equation}
\left\Vert W^{\ast }\right\Vert _{2}^{2}=\sum_{t=1}^{d}n_{t}(W^{\ast
}(N_{t})/n_{t})^{2}\geq n_{\min }\sum_{t=1}^{d}(W^{\ast }(N_{t})/n_{t})^{2}%
\overset{\text{(*)}}{\geq }n_{\min }/d,  \label{cs}
\end{equation}%
where (*) is due to the Cauchy-Schwarz inequality and (\ref{normalization}).
Also, if $t$ maximizes $W^{\ast }(N_{t})/n_{t}$ then
\begin{equation}
\left\Vert W^{\ast }\right\Vert _{2}^{2}\geq n_{t}(W^{\ast
}(N_{t})/n_{t})^{2}=n_{t}\left( \max_{i\in N}W_{i}^{\ast }\right) ^{2}\geq
n_{\min }\left( \max_{i\in N}W_{i}^{\ast }\right) ^{2},  \label{w*norm}
\end{equation}%
and so for $\rho $ in the statement of Lemma \ref{lem:be} the following
holds:

\begin{equation}
\rho =\max_{i\in N}W_{i}^{\ast }/\left\Vert W^{\ast }\right\Vert _{2}\overset%
{\text{(\ref{w*norm})}}{\leq }1/\sqrt{n_{\min }}.  \label{tau}
\end{equation}%
The inequality (\ref{delta bounded by epsilon*}) now follows from Lemma \ref%
{lem:be}, applied to the interval $I=[q^{\ast },q^{\ast }+\tau ^{+}]$:%
\begin{eqnarray*}
\delta (G)\overset{\text{(\ref{deltag prob})}}{\leq }\mathbf{\Pr }(W^{\ast }(%
\mathbf{A}) &\in &[q^{\ast },q^{\ast }+\tau ^{+}]) \\
&&\overset{\text{(\ref{lemma be})}}{\leq }\sqrt{2/\pi }\,\tau
^{+}/\left\Vert W^{\ast }\right\Vert _{2}+2C_{0}\rho \\
&&\overset{\text{(\ref{cs}),(\ref{tau})}}{\leq }\sqrt{2/\pi }\,\frac{\sqrt{d}%
}{\sqrt{n_{\min }}}\max \{\tau (G),0\}+\frac{2C_{0}}{\sqrt{n_{\min }}} \\
&&\overset{\text{(\ref{tau g2})}}{\leq }\,\frac{\sqrt{2d/\pi }(\log
_{2}d+2)+2C_{0}}{\sqrt{n_{\min }}}.
\end{eqnarray*}

\subsection{Proof of Proposition \protect\ref{misclassification for epsilon*}%
}

\label{sec proof of equality}

Fix some $w^{\ast }\in \Delta ^{\ast }(G)$ that is constant on each player
class (which exists by Proposition \ref{lem:minimax}(i)), and let $q^{\ast
}= $ $\min_{A\in \mathcal{W}}w^{\ast }(A).$ This weight system classifies
every winning coalition correctly, and a losing coalition $A$ is
misclassified exactly when $w^{\ast }(A)\geq q^{\ast }$. Because $w^{\ast
}(A)\leq $ $\max_{B\in \mathcal{L}}w^{\ast }(B)=$ $q^{\ast }+\varepsilon
^{\ast }(G)$, every misclassified $A$ must satisfy $w^{\ast }(A)\in $ $%
[q^{\ast },q^{\ast }+\varepsilon ^{+}]$, where $\varepsilon ^{+}=$ $\max
\{\varepsilon ^{\ast }(G),0\}$. Hence, for a random $\mathbf{A}$ with the
uniform distribution on $2^{N},$%
\begin{equation}
\delta ^{\ast }(G)\leq \mathbf{\Pr }(w^{\ast }(\mathbf{A})\in \lbrack
q^{\ast },q^{\ast }+\varepsilon ^{+}]).  \label{deltag prob 1}
\end{equation}

By the Cauchy--Schwarz inequality,
\begin{equation}
\left\Vert w^{\ast }\right\Vert _{2}\geq 1/\sqrt{n}.  \label{cs1}
\end{equation}%
Also, just as in the proof of (\ref{tau}) in Section \ref{proof of prop on
misclassification},%
\begin{equation}
\rho =\max_{i\in N}w_{i}^{\ast }/\left\Vert w^{\ast }\right\Vert _{2}\leq 1/%
\sqrt{n_{\min }}.  \label{tau1}
\end{equation}
The inequality (\ref{upper bound delta(G)}) now follows from Lemma \ref%
{lem:be}, applied to the interval $I=[q^{\ast },q^{\ast }+\varepsilon ^{+}]$:%
\begin{eqnarray*}
\delta ^{\ast }(G)\overset{\text{(\ref{deltag prob 1})}}{\leq }\mathbf{\Pr }%
(w^{\ast }(\mathbf{A}) &\in &[q^{\ast },q^{\ast }+\varepsilon ^{+}]) \\
&&\overset{\text{(\ref{lemma be})}}{\leq }\sqrt{2/\pi }\,\varepsilon
^{+}/\left\Vert w^{\ast }\right\Vert _{2}+2C_{0}\rho \\
&&\overset{\text{(\ref{cs1}),(\ref{tau1})}}{\leq }\sqrt{2/\pi }\,\sqrt{n}%
\max \{\varepsilon ^{\ast }(G),0\}+\frac{2C_{0}}{\sqrt{n_{\min }}} \\
&&\overset{\text{(\ref{epsilon star G})}}{\leq }\sqrt{2/\pi }\,(\log _{2}d+2)%
\frac{\sqrt{n}}{n_{\min }}+\frac{2C_{0}}{\sqrt{n_{\min }}}.
\end{eqnarray*}

\subsection{Proof of Proposition \protect\ref{O(1/n^2)}}

Observe that, for any $A\subseteq N,$%
\begin{equation*}
W\left( A\right) \leq q-2\text{ }\Longleftrightarrow \text{ }W\left(
N\diagdown A\right) \geq q+1
\end{equation*}%
and
\begin{equation*}
W\left( A\right) =q-1\text{ }\Longleftrightarrow \text{ }W\left( N\diagdown
A\right) =q.
\end{equation*}%
In both cases, exactly one coalition in any pair $\{A,N\diagdown A\}$ is
winning, and so $G$ is constant-sum.

By Proposition \ref{prop gappoed weighted games}, $G$ is also $C$-complete,
and $\varepsilon ^{\ast }(G)\leq $ $1/W(N).$ Denote by\ $W_{\max }$ the
maximal individual weight. Since, clearly, $W(N)\geq $ $n_{t}W_{\max }$ for
the class $N_{t}$ with the highest individual weight, we have
\begin{equation}
\varepsilon ^{\ast }(G)\leq 1/W(N)\leq 1/\left( n_{t}W_{\max }\right) \leq
1/\left( n_{\min }W_{\max }\right) .
\label{epsilon star in constant sum game}
\end{equation}

We next show that $n_{\min }<6W_{\max },$ by exploring the implication of
the reverse inequality in the following lemma. Denote by $w_{t}$ the common
(by assumption) weight of every player in $N_{t},$ for each $1\leq t\leq d.$

\begin{lemma}[\textbf{Middle-third profile}]
\label{middle third tie}Suppose that $n_{\min }\geq 6W_{\max }$. Then there
exists $x\in \Lambda \left( \overline{n}\right) $ with $n_{t}/3\leq
x_{t}\leq 2n_{t}/3$ for every $t$ and $\sum_{t=1}^{d}w_{t}x_{t}$ $=q-1.$
\end{lemma}

\begin{proof}
By (iii) in Definition \ref{def punctured games}, there exists $A$ with $%
W(A)=q-1.$ Thus, for $y=s(A),$ we have
\begin{equation*}
\sum_{t=1}^{d}w_{t}y_{t}=W(A)=q-1=\left( W(N)-1\right)
/2=\sum_{t=1}^{d}w_{t}n_{t}/2-1/2,
\end{equation*}%
which leads to:
\begin{equation}
\sum_{t=1}^{d}w_{t}\left( y_{t}-n_{t}/2\right) =-1/2.  \label{kokhav}
\end{equation}

Next, call $y_{t}$ \emph{low} if $y_{t}<n_{t}/3$ and \emph{high} if $%
y_{t}>2n_{t}/3$, and let
\begin{equation*}
\digamma (y)=\sum_{t:\ y_{t}\text{ }\mathrm{is}\text{ }\mathrm{low}}\left(
n_{t}/3-y_{t}\right) +\sum_{t:\text{ }y_{t}\text{ }\mathrm{is}\text{ }%
\mathrm{high}}\left( y_{t}-2n_{t}/3\ \right) \in \ \{0,1/3,2/3,...,n\}.
\end{equation*}%
Then $\digamma (y)=0$ iff $n_{t}/3\leq y_{t}\leq 2n_{t}/3$ for every $t$. We
claim that as long as $\digamma (y)>0$ there is a change $y\mapsto \overline{%
y}=y+w_{r}e_{t}-w_{t}e_{r}$ (that consists of raising some $y_{t}$ by $w_{r}$
and lowering $y_{r}$ by $w_{t}$, thus preserving the equality (\ref{kokhav}%
)), which lowers $\digamma $. Since $\digamma $ has finitely many values,
after finitely many moves we will end with $\digamma =0$. By assumption, $%
W_{\max }\leq n_{t}/6$ for all $t$ throughout.

Assume first that there exists a low $y_{t}$ (namely, $y_{t}<n_{t}/3$)$.$
Then%
\begin{equation*}
w_{t}(y_{t}-n_{t}/2)<w_{t}(n_{t}/3-n_{t}/2)=-w_{t}n_{t}/6\leq -4
\end{equation*}%
because $w_{t}\geq 2$ and $n_{t}\geq 6w_{t}$, and so, by (\ref{kokhav}),%
\begin{equation*}
\sum_{r\neq t}w_{r}(y_{r}-n_{r}/2)>-1/2+4>0.
\end{equation*}%
It follows that for some $r\neq t$ we must have $y_{r}>n_{r}/2$. Consider
the change $y\mapsto $ $\overline{y}=y+w_{r}e_{t}-w_{t}e_{r}$. Clearly,
\begin{equation*}
\overline{y}_{t}=y_{t}+w_{r}<n_{t}/3+w_{r}\leq n_{t}/3+n_{t}/6=n_{t}/2
\end{equation*}%
and
\begin{equation*}
\overline{y}_{r}=y_{r}-w_{t}>n_{r}/2-w_{t}\geq n_{r}/2-n_{r}/6=n_{r}/3;
\end{equation*}%
$\overline{y}$ and $y$ are identical in coordinates other than $r$ and $t.$
Thus, $\overline{y}\in \Lambda \left( \overline{n}\right) ,$ $\overline{y}%
_{t}$ is not high and $\overline{y}_{r}$ is not low, and so
\begin{eqnarray*}
&&\digamma (y)-\digamma (\overline{y}) \\
&=&\left( n_{t}/3-y_{t}-\max \{n_{t}/3-y_{t}-w_{r},0\}\right) \\
&&{}+\left( \max \{y_{r}-2n_{r}/3,0\}-\max \{y_{r}-w_{t}-2n_{r}/3,0\}\right)
.
\end{eqnarray*}%
The first summand above is positive because $\frac{n_{t}}{3}-y_{t}>0$ by
assumption, and the second summand is nonnegative. Thus, the transition from
$y$ to $\overline{y}$ lowers $\digamma $.

Assume next that there is no low $y_{t},$ but there is some high $y_{t}$
(namely, $y_{t}>2n_{t}/3$)\textit{.} Then $w_{t}(y_{t}-n_{t}/2)>w_{t}n_{t}/6%
\geq 4$, and so, by (\ref{kokhav}),%
\begin{equation*}
\sum_{r\neq t}w_{r}(y_{r}-n_{r}/2)<-1/2-4<0.
\end{equation*}%
It follows that for some $r\neq t$ we must have $y_{r}<n_{r}/2$ (but by
assumption, this $y_{r}$ is not low). Consider the change $y\mapsto $ $%
\overline{y}=y-w_{r}e_{t}+w_{t}e_{r}$. Clearly,
\begin{equation*}
\overline{y}_{t}=y_{t}-w_{r}>2n_{t}/3-w_{r}\geq 2n_{t}/3-n_{t}/6=n_{t}/2
\end{equation*}%
and
\begin{equation*}
\overline{y}_{r}=y_{r}+w_{t}<n_{r}/2+w_{t}\leq n_{r}/2+n_{r}/6=2n_{r}/3;
\end{equation*}%
$\overline{y}$ and $y$ are identical in coordinates other than $r$ and $t.$
Thus, $\overline{y}\in \Lambda \left( \overline{n}\right) ,$ its $\overline{y%
}_{t}$ is not low and $\overline{y}_{r}$ is not high (nor low), and so
\begin{equation*}
\digamma (y)-\digamma (\overline{y})=y_{t}-2n_{t}/3-\max
\{y_{t}-w_{r}-2n_{t}/3,0\}>0
\end{equation*}%
because $y_{t}-\frac{2n_{t}}{3}>0$ by assumption. We conclude that the
transition from $y$ to $\overline{y}$ lowers $\digamma $ in this case as
well.

Thus, starting from a profile $y=s(A),$ for some $A$ with $W(A)=q-1,$ after
a finite sequence of changes in the profile we obtain $x\in \Lambda \left(
\overline{n}\right) $ for which: (i) (\ref{kokhav}) holds for $y=x,$
implying that $\sum_{1\leq t\leq d}w_{t}x_{t}=q-1;$ and (ii) $\digamma (x)=0$%
, implying that $n_{t}/3\leq x_{t}\leq 2n_{t}/3$ for every $t$ by the
definition of $\digamma .$
\end{proof}

If $n_{\min }\geq 6W_{\max }$, consider a profile $x$ as guaranteed in Lemma %
\ref{middle third tie}, and let $v=3x-\overline{n}.$ By the properties of $%
x, $ clearly $v\in \Lambda \left( \overline{n}\right) $ and
\begin{equation*}
\sum_{t=1}^{d}w_{t}v_{t}=\sum_{t=1}^{d}w_{t}\left( 3x_{t}\right)
-W(N)=3(q-1)-(2q-1)=q-2.
\end{equation*}%
Thus, there exists $A\subset N$ with $W(A)=q-2,$ but this contradicts (iv)
in Definition \ref{def punctured games}. We conclude that $n_{\min
}<6W_{\max }.$ Therefore%
\begin{equation*}
\varepsilon ^{\ast }(G)\overset{\text{(\ref{epsilon star in constant sum
game})}}{\leq }1/\left( n_{\min }W_{\max }\right) <6/n_{\min }^{2}.
\end{equation*}

\end{document}